\documentclass[pra,twocolumn, preprintnumbers,amsmath,amssymb, nofootinbib, accepted=2018-07-09]{quantumarticle}

\usepackage{graphicx, color, graphpap}
\usepackage{amsmath}
\usepackage{enumitem}
\usepackage{amssymb}
\usepackage{amsthm}
\usepackage{pstricks}
\usepackage{float}
\usepackage{multirow}
\usepackage{hyperref}
\usepackage{overpic}
\usepackage[T1]{fontenc}
\usepackage{algorithm}
\usepackage{algpseudocode}


\long\def\ca#1\cb{} 

\newcommand{\abs}[2][]{#1| #2 #1|}

\newcommand{\norm}[2][]{#1| \! #1| #2 #1| \! #1|}

\newcommand{\ket}[1]{|#1\rangle}               
\newcommand{\bra}[1]{\langle #1|}              
\newcommand{\dya}[1]{\ket{#1}\!\bra{#1}}
\newcommand{\dyad}[2]{\ket{#1}\!\bra{#2}}        
\newcommand{\ipa}[2]{\langle #1,#2\rangle}      

\newcommand{\AC}{\mathcal{A}}

\newcommand{\DC}{\mathcal{D}}
\newcommand{\EC}{\mathcal{E}}

\newcommand{\GC}{\mathcal{G}}

\newcommand{\IC}{\mathcal{I}}

\newcommand{\ZC}{\mathcal{Z}}

\newcommand{\Tr}{{\rm Tr}}

\newcommand{\pass}{\text{pass}}
\newcommand{\ave}[1]{\langle #1\rangle}               
\renewcommand{\geq}{\geqslant}
\renewcommand{\leq}{\leqslant}

\newcommand{\mted}[3]{\langle#1|#2|#3\rangle }

\newcommand{\leak}{\text{leak}^{\text{EC}}_{\text{obs}}}
\newcommand{\vectext}{\text{vec}}

\newcommand{\At}{\widetilde{A}}
\newcommand{\Bt}{\widetilde{B}}

\newcommand{\Ab}{\overline{A}}
\newcommand{\Bb}{\overline{B}}

\newcommand{\ot}{\otimes}
\newcommand{\ad}{^\dagger}

\newcommand*{\id}{\openone}

\newcommand{\rhot}{\tilde{\rho}}

\newcommand{\al}{\alpha }

\newcommand{\ep}{\epsilon}

\newcommand{\St}{\tilde{\mathbf{S} }}
\newcommand{\Sh}{\hat{\mathbf{S} }}

\newtheoremstyle{example}{\topsep}{\topsep}%
{}
{}
{\bfseries}
{:}
{   }
{\thmname{#1}\thmnumber{ #2}}
\theoremstyle{example}

\newtheorem{theorem}{Theorem}
\newtheorem{lemma}{Lemma}

\newtheorem{proposition}{Proposition}

\theoremstyle{definition}

\begin{document}

\title{Reliable numerical key rates for quantum key distribution}

\author{Adam Winick}
\affiliation{Institute for Quantum Computing and Department of Physics and Astronomy, University of Waterloo, N2L3G1 Waterloo, Ontario, Canada}

\author{Norbert L\"utkenhaus}
\affiliation{Institute for Quantum Computing and Department of Physics and Astronomy, University of Waterloo, N2L3G1 Waterloo, Ontario, Canada}

\author{Patrick J. Coles}
\thanks{Current affiliation: Theoretical Division, Los Alamos National Laboratory, Los Alamos, NM 87545, USA.}
\affiliation{Institute for Quantum Computing and Department of Physics and Astronomy, University of Waterloo, N2L3G1 Waterloo, Ontario, Canada}

\begin{abstract}
In this work, we present a reliable, efficient, and tight numerical method for calculating key rates for finite-dimensional quantum key distribution (QKD) protocols. We illustrate our approach by finding higher key rates than those previously reported in the literature for several interesting scenarios (e.g., the Trojan-horse attack and the phase-coherent BB84 protocol). Our method will ultimately improve our ability to automate key rate calculations and, hence, to develop a user-friendly software package that could be used widely by QKD researchers. 
\end{abstract}

\maketitle

\section{Introduction}

The possibility of large-scale quantum computers in the near future has spawned the field of quantum-safe cryptography \cite{Campagna2015}. This includes classical techniques based on computational hardness as well as information-theoretic approaches via physical assumptions. The latter invokes either assumptions about the physical channel \cite{Wyner1975}, or less restrictive, the assumption only that the laws of quantum physics are correct, which is the basis for quantum key distribution (QKD). See Ref.~\cite{Scarani2009} for a review of QKD and Ref.~\cite{Lo2014} for an update of recent progress. Both classical methods and QKD will likely play a role in quantum-safe cryptographic implementations.

The maturity of QKD technology is evidenced by a recent QKD satellite launch \cite{Xin2011} as well as developments of fiber-based networks \cite{Peev2009, Sasaki2011, Wang2014}, suggesting that global networks are on the horizon. Still, there remains important open problems in QKD theory, such as (1) optimizing the practicality of QKD protocols to make them easily implementable, and (2) understanding the effects of device imperfections and side channels. 

Solving these problems requires a robust theoretical method for evaluating a QKD protocol's performance, which involves a detailed security analysis. Performance is then quantified by the key rate - the number of bits of secret key obtained per exchange of quantum signal. Unfortunately, analytical methods for calculating the key rate are highly technical, are often limited in scope to particular protocols, and invoke inequalities that introduce looseness into the calculation.

We therefore focus our efforts on numerical methods, which are inherently more robust to both device imperfections and changes in protocol structure. Furthermore, numerics can be made user-friendly, such that the user needs only to define the specifications of the protocol and then the computer performs the key rate calculation. As an example, our group recently released a software package for this purpose.\footnote{This software can be downloaded from the website: https://lutkenhausgroup.wordpress.com/qkd-software/.} 

The key rate calculation involves minimizing a convex function over all eavesdropping attacks that are consistent with the experimental data \cite{Renner2005,Renner2005a, Watanabe2008, Matsumoto2013}. When employing numerics, one issue that arises is the efficiency of this optimization. This issue is particularly important for high-dimensional QKD protocols, or protocols with many signal states, since the relevant optimization involves many parameters. In such cases, the computational time can be very long - sometimes days - so it is crucial to implement a high efficiency algorithm. 
 
Another issue with numerics is reliability. This is more subtle but also more important than the efficiency issue. Due to the inherent paranoia in cryptography, it is natural to ask whether numerically calculated key rates are trustworthy. After all, computers have finite numerical precision. Furthermore, optimization algorithms never truly reach the global optimum, as termination conditions always have some non-zero tolerance. Since QKD is now a serious, real-world technology, key rates must come with a security guarantee, and hand-waving at these numerical issues will not suffice. 

In this work, we present a numerical method that solves both the reliability and efficiency issues. Our method provides reliable lower bounds on the key rate with arbitrary tightness for finite-dimensional QKD protocols. Furthermore it is highly efficient and typically returns a key rate within seconds or less on one's personal computer. To illustrate our method, we apply it to three practically interesting scenarios. Namely we consider the Trojan-horse attack~\cite{Vakhitov2001, Gisin2006, Lucamarini2015}, the BB84 protocol with phase-coherent signal states \cite{Huttner1995,Lo2007}, and the BB84 protocol with detector efficiency mismatch \cite{Fung2009}. We improve upon literature key rates in all three cases.

Directly calculating the key rate involves a minimization problem (see Sec.~\ref{sctbackground}). When solving this on a computer, the algorithm will terminate before reaching the global optimum and hence will return an upper bound on the true key rate. However, we are interested in reliable lower bounds, i.e., achievable key rates. In previous work \cite{Coles2016}, we noted this issue as motivation for transforming the optimization problem to the so-called dual problem \cite{Boyd2010}. This transforms the minimization problem into a maximization problem. Therefore, the dual problem will return a lower bound on the key rate, as desired. This method led to novel insights for particular protocols as discussed in \cite{Coles2016}. However, in order to simplify the optimization in the dual problem, we invoked an inequality that in some cases introduces looseness into the key rate and, furthermore, makes the optimization problem non-convex. Ultimately the non-convexity reduces the efficiency of this approach, making it difficult to apply to protocols with large numbers of signal states.

We therefore present a method here that retains the efficiency of convex optimization, but also has the reliability of the dual problem. Our approach is to break up the calculation into two steps. The first step approximately minimizes the convex function, and hence finds an eavesdropping attack that is close to optimal. The second step takes this approximately optimal attack and converts it into a lower bound on the key rate. Breaking it up into these two steps adds flexibility to our method, in that any algorithm can be employed for the initial minimization of the convex function.

Our main technical result is to provide a recipe for performing the second step, i.e., for converting a near-optimal attack into a tight lower bound on the key rate. At the technical level, we derive our main result first by linearizing the problem and then by transforming to the dual problem of the subsequent linearized problem. The idea is that, for a convex function, any linearization (about any point) will undercut (and hence lower bound) the curve. One obtains the tightest lower bound by this method if one linearizes about a point corresponding to the global minimum of the convex function.

We emphasize several differences between this work and Ref.~\cite{Coles2016}. Our method presented herein is tight, essentially giving the optimal key rate, whereas Ref.~\cite{Coles2016} invoked the Golden-Thompson inequality which for certain protocols introduces looseness into the calculated key rates. This difference can be seen below in our Figures \ref{fgrEffMismatch}, \ref{fgrTrojan}, and \ref{fgrLP}, which show a clear gap between our key rates and those computed from the method in Ref.~\cite{Coles2016}. Another difference, as noted earlier, is the lack of convexity of the method in Ref.~\cite{Coles2016}. For a small number of optimization parameters, e.g., as encountered with entanglement-based QKD protocols, the non-convex method in Ref.~\cite{Coles2016} works quite well. On the other hand, for the method in Ref.~\cite{Coles2016}, prepare-and-measure QKD protocols require adding a number of optimization parameters that scale quadradically in the number of signal states, as we need to describe the pairwise quantum-mechanical overlaps between the signals (see Ref.~\cite{Coles2016} and the discussion below around Eq.~\eqref{eqnSourceReplace95}), which can lead to long computation times. In contrast, the convexity of the method presented herein leads to more efficient scaling of the computation time with the number of signal states for prepare-and-measure QKD, despite the larger amount of open parameters of this method.

In what follows we first give background on key rate calculations in the next section. Then we present our main result in Sec.~\ref{sctmainresult}. In Sec.~\ref{sctqkdframework} we describe how our approach applies to a general class of QKD protocols. We illustrate our method for three interesting example protocols in Sec.~\ref{sctexamples}, and finally we conclude in Sec.~\ref{sctconclusion}. Technical details can be found in the Appendix.

\section{Background}\label{sctbackground}

The well-known asymptotic key rate formula \cite{Devetak2005} is given by the difference of two information-theoretic quantities associated, respectively, with privacy amplification (PA) and error correction (EC). These two terms appear in the following expression for the key rate per signal
\begin{align}
\label{eqnPAEC3}
K &= p_{\pass}\left(\min_{\rho \in\mathbf{S} } \widehat{f}(\rho) -  \leak \right) \\
\label{eqnPAEC3_2}
&= \left(\min_{\rho \in\mathbf{S} } f(\rho)\right)  - p_{\pass}\cdot \leak\,.
\end{align}
We explain this expression in more detail in Sec.~\ref{sctqkdframework}. For now, we note that $p_{\pass}$ refers to the probability for passing the post-selection (e.g., sifting) in the protocol, $\leak$ 
denotes the number of bits of information (per signal that passed post-selection) that Alice publicly reveals during error correction, and $f(\rho) = p_{\pass}\cdot\widehat{f}(\rho)$ is a function defined below in Eq.~\eqref{eqnPAEC4b}.

The first term in \eqref{eqnPAEC3} is the PA term. Here, $\rho$ is the density operator shared by Alice, Bob, and possibly other parties involved in the protocol. (Note that prepare-and-measure protocols can be recast as entanglement-based protocols and are described by same mathematics, see Sec.~\ref{sctqkdframework} for elaboration.) We assume that the state has an i.i.d.\ (independent, identically distributed) structure, and hence it makes sense to discuss the state $\rho$ associated with a single round of quantum communication. This i.i.d.\ structure corresponds to Eve doing a so-called collective attack. However, the security of our derived asymptotic key rate also holds against the most general attacks (coherent attacks) if one imposes that the protocol involves a random permutation of the rounds (a symmetrization step) such that the Quantum de Finetti theorem \cite{Renner2005, Renner2007} applies.

The density operator $\rho$ is unknown, but the asymptotic experimental data gives linear constraints on it, of the form
\begin{align} \label{eqnconstraints1234}
\Tr(\Gamma_i\rho) = \gamma_i, \quad\forall i\,,
\end{align}
where the $\Gamma_i$ are Hermitian operators. Let $\mathbf{S}$ denote the set of states that satisfy these constraints
\begin{align} \label{eqnconstraintSet}
\mathbf{S} = \{\rho\in \mathbf{H_{+}} \mid \Tr(\Gamma_i \rho) = \gamma_i, \forall i\}\,,
\end{align}
where $\mathbf{H_{+}}$ is the set of positive semidefinite operators. Also, we add the identity to the set $\{\Gamma_i\}$ to enforce that $\Tr(\rho ) = 1$, giving a total of $n$ constraints.

The key rate calculation is an optimization problem, since we must consider the worst-case scenario (the most powerful eavesdropping attack) that is consistent with the experimental data. Hence Eq.~\eqref{eqnPAEC3} involves minimizing over all $\rho \in\mathbf{S}$. Note that the error correction term is exactly determined by the observations, and hence we can pull it out of the optimization in \eqref{eqnPAEC3_2}. Similarly, $p_{\pass}$ is known from the observations and is pulled into the optimization in \eqref{eqnPAEC3_2}.

As discussed in Sec.~\ref{sctqkdframework}, $f(\rho)$ can be written as
\begin{align} \label{eqnPAEC4b}
f(\rho) = D(\hspace{2pt}   \GC(\rho) || \ZC(\GC(\rho))      \hspace{2pt})\,,
\end{align}
where $D(\sigma ||\tau):= \Tr(\sigma \log \sigma) -\Tr(\sigma \log \tau)$ is the relative entropy, $\GC$ is a completely positive (CP) map, and $\ZC$ is a completely positive trace preserving (CPTP) map, more specifically a pinching quantum channel. (Sec.~\ref{sctqkdframework} discusses the meaning of $\GC$ and $\ZC$, which are respectively related to the post-selection and the key map of the QKD protocol). Due to the joint convexity of the relative entropy and the fact that $\GC$ and $\ZC$ are linear maps, the function $f(\rho)$ is convex in $\rho$. Furthermore the problem
\begin{align} \label{eqnPAEC4}
\alpha :=\min_{\rho \in\mathbf{S} } f(\rho) 
\end{align}
is a convex optimization problem since the set $\mathbf{S}$ is convex (see, e.g., Ref.~\cite{Boyd2010}). While efficient numerical methods are known for such convex problems, the key rate calculation is unique compared to other convex problems, in that getting ``close'' to the optimal point is not good enough. One needs a reliable lower bound on the key rate, i.e., guaranteed security.

\section{Main result}\label{sctmainresult}

\subsection{Reliable lower bound}\label{sctmainresult1}

We now show how to lower bound the minimization problem in \eqref{eqnPAEC4}. Our strategy is to break up the key rate calculation into two steps:
\begin{itemize}
\item Step 1: Find an eavesdropping attack that is close to optimal, which gives an upper bound on the key rate.
\item Step 2: Convert this upper bound to a lower bound on the key rate.
\end{itemize}

With our approach, Step~1 does not need to be perfect - any eavesdropping attack may be used as an input for Step~2. However, if Step~1 returns the optimal attack, our lower bound calculated by Step 2 will be tight. Furthermore, our method for Step 2 is continuous around the optimal attack. Thus, finding a near-optimal attack produces a near-optimal lower bound.

Step~1 may be solved in various ways using convex optimization methods \cite{Boyd2010}. For concreteness, Sec.~\ref{sctalgorithm} presents one such method, which exploits the structure of our problem and is relatively fast. 

On the other hand, our main result is a method for performing Step 2. We approach Step 2 via a sequence of theorems that successively improve the reliability and robustness of the lower bounds which they return. First, Theorem~\ref{thm1} presents the conceptual foundation for our lower bounding method. However, this theorem is stated under a restrictive assumption that may not hold in special cases. Therefore, we extend our result in Theorem~\ref{thm2}. Finally, we improve our result once more in Theorem~\ref{thm3}, which addresses numerical imprecision and is directly useful for numerical key rate calculations. Sec.~\ref{sctTightness} presents the argument that our method yields arbitrarily tight bounds on the key rate.

We now present our main result in its simplest conceptual form. To state this result, we first define the gradient of $f$ at point $\rho$, whose representation in the standard basis $\{\ket{j}\}$ is
\begin{align}
\label{eqnGradDef836}
\nabla f(\rho):=\sum_{j,k} d_{jk}\dyad{j}{k},\quad\text{with   }d_{jk}:=\frac{\partial f(\sigma)}{\partial \sigma_{jk}} \Bigg | _{\sigma=\rho} 
\end{align}
and $\sigma_{jk}:= \mted{j}{\sigma}{k}$.

\bigskip
\begin{theorem}
\label{thm1}
Given any $\rho \in \mathbf{S}$, if $\nabla f(\rho)$ exists, then
\begin{align}
\label{eq:DualLowerBound}
\alpha \geq \beta(\rho)\,,
\end{align}
where $\al$ was defined in \eqref{eqnPAEC4} and
\begin{align}
\label{eq:DualLowerBound2}
\beta(\sigma) &:= f(\sigma)-\Tr(\sigma^T \nabla f(\sigma))+\max_{\vec{y} \in \mathbf{S}^*(\sigma)} \vec{\gamma} \cdot \vec{y}\,, \\
\label{eq:DualLowerBound3}
\mathbf{S^*}(\sigma) &:= \left\{ \vec{y} \in \mathbb{R}^n \mid \sum_i y_i \Gamma_i^T \leq \nabla f(\sigma) \right\}\,.
\end{align}
Here, the transpose $^T$ is taken in the same basis as that used to define the gradient in \eqref{eqnGradDef836}, and $\vec{\gamma} = \{\gamma_i\}$ is the vector of expectation values from \eqref{eqnconstraints1234}, which is of size $n$ (the total number of constraints). Furthermore, equality in \eqref{eq:DualLowerBound} holds if $f(\rho) = \alpha$, i.e., if $\rho$ corresponds to an optimal attack.
\end{theorem}
\bigskip

The proof is given in Appendix~\ref{sctAppProof}. It involves linearizing the convex function $f$ about point $\rho$ and then transforming the subsequent linearized problem to its dual problem (see~\cite{Boyd2010} for discussion of duality).

Figure~\ref{fgrMainResult} illustrates the basic idea of Theorem~\ref{thm1}. Theorem~\ref{thm1} takes any feasible eavesdropping attack $\rho$, which gives an upper bound $f(\rho)$ on $\alpha$, and converts it into a reliable lower bound on $\alpha$. The fact that \eqref{eq:DualLowerBound2} involves a maximization is crucial for the reliability of the calculation. Since maximization involves approaching the solution from below, every number that the computer outputs is a lower bound on $\alpha$, even if the computer does not reach the global maximum.

\begin{figure}[tbp]
\begin{center}
\vspace{0.2cm}
\begin{overpic}[width=8.1cm]{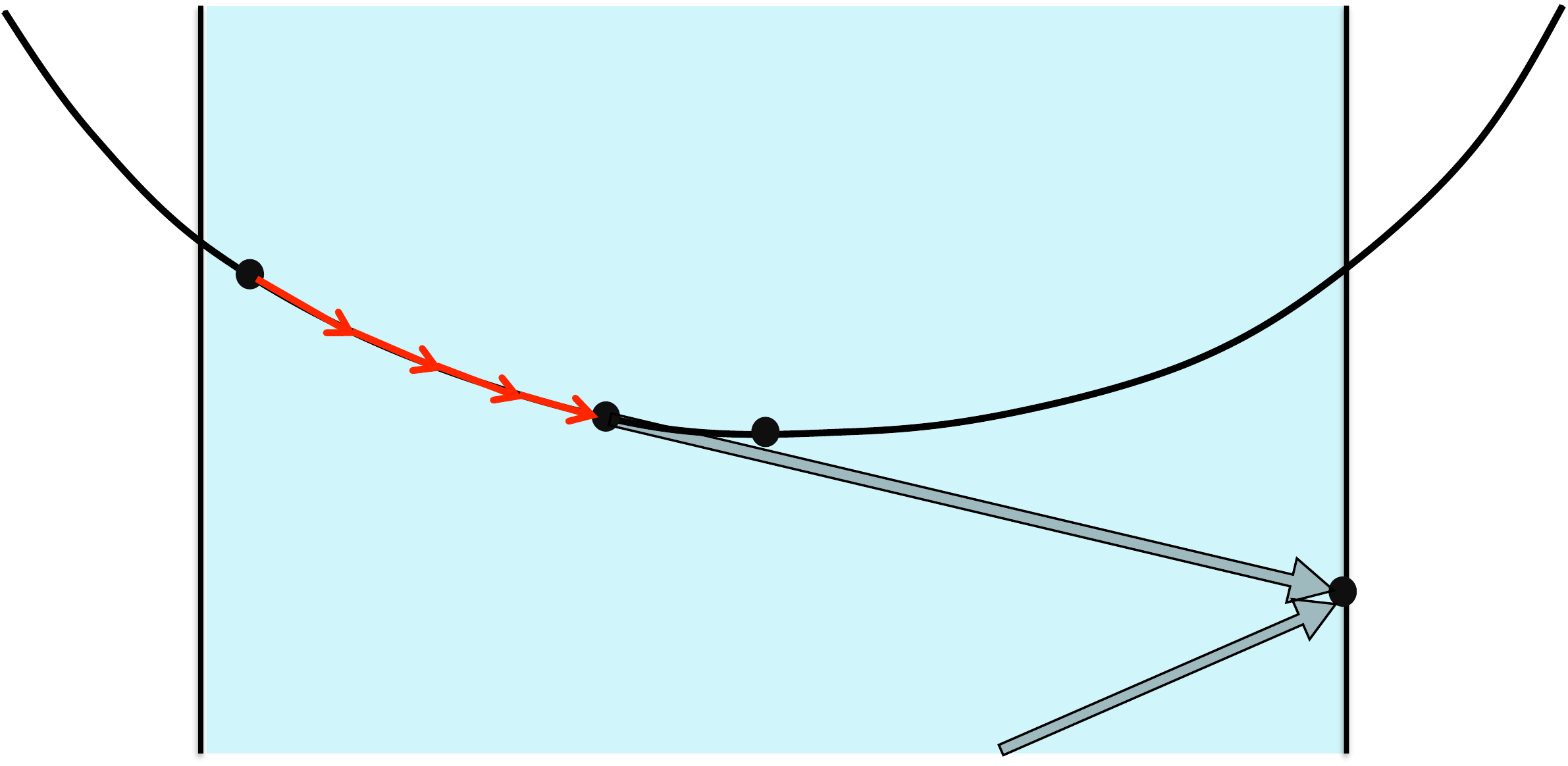}
\put(2,37){\normalsize $f$}
\put(15,33.9){\normalsize $\rho_0$}
\put(38,25){\normalsize $\rho$}
\put(48,23.9){\normalsize $\rho^*$}
\put(79,43.9){\normalsize $\mathbf{S}$}
\put(23,29.2){\normalsize Step 1}
\put(42,9.2){\normalsize Step 2}
\put(87,11.9){\normalsize Our lower}
\put(88,8){\normalsize bound}
\put(58,18.5){\rotatebox{-13}{\normalsize Linearization}}
\put(65,3.5){\rotatebox{22}{\normalsize Dual}}
\end{overpic}
\caption{Illustration of our lower-bounding method. Step~1 is any algorithm that takes an initial feasible point $\rho_0$ and outputs another feasible point $\rho$, which may or may not be close to the optimal attack $\rho^*$. Note that $f(\rho)$ provides an upper bound on $f(\rho^*)$ and, hence, on the key rate. Step~2 converts this into a lower bound, by solving the dual problem of the linearization of $f$ about point $\rho$. Since the linearization undercuts the curve $f$ and since the dual problem is a maximization, our lower bound is reliable even if the numerical calculation does not reach the global optimum.}
\label{fgrMainResult}
\end{center}
\end{figure}
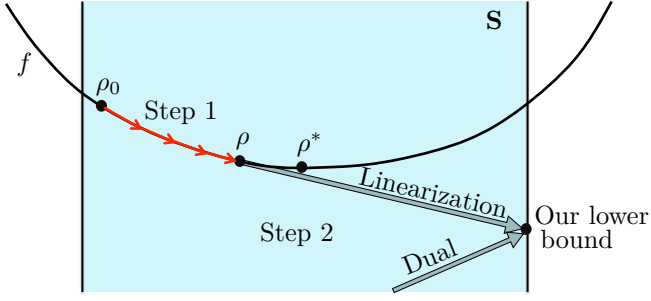

\subsection{A more robust bound}\label{sctmainresult2}

The assumption that $\nabla f$ exists over the entirety of $\mathbf{S}$ does not necessarily hold. In particular, the gradient has the form
\begin{align}
\label{eqngradient1}
\left[ \nabla f(\rho) \right]^T = \GC^\dagger (\log \GC(\rho)) - \GC^\dagger (\log \ZC(\GC(\rho))) \,,
\end{align}
which is derived using rules for derivatives of matrix traces \cite{Berger2007}. (For example, the derivative of $\Tr (G(X))$ is $g(X)^T$ where $g(\cdot)$ is the scalar derivative of $G(\cdot)$.) By inspection, we see that if $\GC(\rho)$ is singular (i.e., not full rank), $\nabla f(\rho)$ may not exist, and hence Theorem~\ref{thm1} would not apply. Ideally we would like a bound that holds for all $\rho \in \mathbf{S}$. Towards this end, it is helpful to introduce a slight perturbation that maps $\GC(\rho)$ to its interior, i.e., points on the boundary become interior points. The key fact that makes this perturbation useful is that we can upper bound the difference between the function $f$ evaluated on the perturbed and unperturbed inputs, and this ultimately leads to our result in Theorem~\ref{thm2} below. We introduce this perturbation as follows.

Let $\epsilon$ be such that $0< \epsilon < 1$, and consider a perturbed channel
\begin{align}
\label{eqnPerturbedChannel}
\GC_{\epsilon}(\rho) := (\DC_{\epsilon}\circ \GC)(\rho)
\end{align}
where $\DC_{\epsilon}$ is the depolarizing channel,
\begin{align}
\label{eqnPerturbedChannel2}
\DC_{\epsilon}(\rho) =(1-\epsilon) \rho + \epsilon \hspace{2 pt} \tau
\end{align}
where $ \tau  = \id / d'$ is the maximally mixed state. Here (and for future reference) we define $d$ and $d'$, respectively, to be the dimension of $\rho$ and $\GC(\rho)$. Furthermore, we define a perturbed function
\begin{align}
\label{eqnPerturbedfunction1}
f_{\epsilon}(\rho) :=D(\GC_{\epsilon}(\rho) || \ZC( \GC_{\epsilon}(\rho)))\,,
\end{align}
which is just \eqref{eqnPAEC4b} with $\GC$ replaced by $\GC_{\epsilon}$. The key point is the following lemma.

\bigskip
\begin{lemma}
\label{lemma_gradient_exists}
The gradient $\nabla f_{\epsilon}(\rho)$, which is given by
\begin{align}
\label{eqngradient1123}
\left[ \nabla f_{\epsilon}(\rho) \right]^T = \GC_{\epsilon}^\dagger (\log \GC_{\epsilon}(\rho)) - \GC_{\epsilon}^\dagger (\log \ZC(\GC_{\epsilon}(\rho))) \,,
\end{align}
always exists for all $\rho \geq 0$.
\end{lemma}
\bigskip

As elaborated in Appendix~\ref{sctAppGradient}, Lemma~\ref{lemma_gradient_exists} follows from the fact that $\GC_{\epsilon}(\rho) > 0$ and similarly $\ZC(\GC_{\epsilon}(\rho)) > 0$. Hence the logarithm of these two matrices in \eqref{eqngradient1123} is well-defined.

The above lemma motivates the following theorem, which is a more robust version of Theorem~\ref{thm1}.

\bigskip
\begin{theorem}
\label{thm2}
Let $\rho \in \mathbf{S}$, where $\rho$ is $d\times d$ and $\GC(\rho)$ is $d' \times d'$. Let $\epsilon \in \mathbb{R}$. If $0 < \epsilon \leq 1/[e(d'-1)]$, where $e$ is the base of the natural logarithm, then
\begin{align}
\label{eq:DualLowerBound5215}
\alpha \geq \beta_{\epsilon}(\rho) - \zeta_\epsilon
\end{align}
where
\begin{align}
\label{eqnBetaEpsilonDef}
\beta_{\epsilon}(\sigma) &:=  f_{\epsilon}(\sigma)-\Tr(\sigma^T \nabla f_{\epsilon}(\sigma))+\max_{\vec{y} \in \mathbf{S}_{\epsilon}^*(\sigma)} \vec{\gamma} \cdot \vec{y}\,, \\
\label{eq:DualLowerBound3234}
\mathbf{S_{\epsilon}^*}(\sigma) &:= \left\{ \vec{y} \in \mathbb{R}^n \mid \sum_i y_i \Gamma_i^T \leq \nabla f_{\epsilon}(\sigma) \right\}\,,\\
\label{eqnDeltaEpsilonDef}
\zeta_{\ep} &:= 2 \epsilon (d' - 1) \log \frac{d'}{\epsilon (d' - 1)} \,.
\end{align}
\end{theorem}
\bigskip

The proof is given in Appendix \ref{sctAppThm2}. The basic idea of Theorem~\ref{thm2} is that it is essentially just Theorem~\ref{thm1}, but applied to a slightly perturbed function $f_{\epsilon}$. Alternatively, one can think of Theorem~\ref{thm2} as being Theorem~\ref{thm1} applied to a perturbed state, $\GC_{\epsilon}(\rho)$ instead of $\GC(\rho)$, where the perturbation maps states on the boundary to the interior. 

Note that Theorem~\ref{thm2} generalizes Theorem~\ref{thm1}, since we have that
\begin{align}
\lim_{\epsilon \to 0+}\beta_{\epsilon}(\rho)-\zeta_\epsilon = \beta (\rho) \,,
\end{align}
assuming that $\nabla f(\rho)$ exists. However, Theorem~\ref{thm2} is more robust than Theorem~\ref{thm1}, because it holds for any $\rho \in \mathbf{S}$, even if $\nabla f(\rho)$ does not exist.

\subsection{Finite precision computation}\label{sctmainresult4}

Finding a lower bound with our method requires a computer program for any nontrivial QKD protocol. The finite precision inherent to computational methods introduces errors that threaten the reliability of the calculated lower bounds.

With this threat in mind, we identify all possible sources of errors in evaluating Theorem~\ref{thm1}. First, the variables will not be precise: $\rho$ will not exactly satisfy the constraints, and $\{\Gamma_i\}$ and $\{\gamma_i\}$ will not be exact. Second, function evaluations will not be precise: every function evaluation introduces errors.

The aforementioned errors fall into broader categories. Variable imprecision is implementation independent since we can characterize the imprecision in a universal manner. Function-evaluation imprecision is implementation dependent since in general, its characterization varies widely with the particular algorithm (particularly for evaluating nontrivial functions such as the matrix logarithm appearing in our objective function). Since the latter kind of errors depend on the implementation, a universal treatment of them is not possible. In principle, it is possible to bound the effect of such errors for a particular implementation \cite{Al-Mohy2012}, although it is beyond the scope of this article. On the other hand, we give a full treatment of implementation-independent errors in what follows.

From a computational perspective, it is virtually impossible to find an element strictly in $\mathbf{S}$. Furthermore, for many problems, the computer representation $\{\tilde{\Gamma}_i\}$ and $\{\tilde{\gamma}_i\}$ do not equal $\{\Gamma_i\}$ and $\{\gamma_i\}$ by the nature of their numerical construction. This lack of strictly constraint-satisfying density matrices motivates the need for a relaxed theorem.

We show (See Appendix~\ref{impreciserepresentations}) that both strict set membership and imprecise constraints can be described by the inequalities
\begin{align}
\label{eqnrelaxedconstraints1}
\abs{\Tr(\tilde{\Gamma}_i \rho)-\tilde{\gamma}_i} \leq \epsilon' \,, \forall i\,,
\end{align}
where $\epsilon' > 0$ is an appropriately chosen number. Determining $\epsilon'$ depends heavily on how $\{\tilde{\Gamma}_i\}$ and $\{\tilde{\gamma}_i\}$ are constructed, so we leave out a precise analysis. In general, given a suitable high-precision computing environment, $\epsilon'$ may be made arbitrarily small. (For example, for our calculations in Sec.~\ref{sctexamples}, we choose $\epsilon' < 10^{-12}$.) The bound on the unknown constraint violations motivates the introduction of a relaxed set
\begin{align}
\label{eqnrelaxedconstraints2}
\St_{\epsilon'} := \left\{\rho\in \mathbf{H_{+}} \middle | \abs{\Tr(\tilde{\Gamma}_i \rho)-\tilde{\gamma}_i} \leq \epsilon', \forall i \right\}\,.
\end{align}
So long as $\epsilon'$ is larger than the constraint violations, the following relation holds (See Appendix~\ref{impreciserepresentations}):
\begin{align}
\label{eqnrelaxedconstraints212342}
\mathbf{S} \subseteq \St_{\epsilon'}\,.
\end{align}
With this new set, we now present a relaxed version of Theorem~\ref{thm2}, with the proof given in Appendix~\ref{app:proofRelaxed}.

\bigskip
\begin{theorem}
\label{thm3}
Let $\rho\in \St_{\epsilon'}$, where $\rho $ is $d\times d$, $\epsilon' > 0$, and $0 < \epsilon \leq 1/[e(d'-1)]$. Then
\begin{align}
\label{eqnthm3mainresult}
\alpha \geq \beta_{\epsilon\epsilon'}(\rho)-\zeta_\epsilon
\end{align}
where $\zeta_\epsilon$ was defined in \eqref{eqnDeltaEpsilonDef} and
\begin{align}
\label{eqnthm3betaeeprime}
\beta_{\epsilon\epsilon'}(\sigma) := &f_{\epsilon}(\sigma)-\Tr(\sigma^T\nabla f_{\epsilon}(\sigma)) \notag\\
&+   \max_{ (\vec{y},\vec{z}) \in \St_{\epsilon}^*(\rho)}  \left( \vec{\tilde{\gamma}} \cdot \vec{y}  - \epsilon' \sum_{i=1}^n z_i \right) \,.
\end{align}
Here, the set $\St_{\epsilon}^*(\sigma)$ is defined by
\begin{align}
&\St_{\epsilon}^*(\sigma) := \notag\\
&\left\{ (\vec{y},\vec{z}) \in (\mathbb{R}^{n}, \mathbb{R}^{n}) \mid -\vec{z} \leq \vec{y} \leq \vec{z} ,\hspace{4pt} \sum_{i=1}^n  y_i   \tilde{\Gamma}_i^T  \leq \nabla f_{\epsilon}(\sigma) \right\} \,.
\end{align}
\end{theorem}
\bigskip

Notice that as $\epsilon'$ goes to zero the last term in \eqref{eqnthm3betaeeprime} vanishes and hence 
\begin{align}
\lim_{\epsilon'\to 0+} \beta_{\epsilon\epsilon'}(\sigma) = \beta_{\epsilon}(\sigma),
\end{align}
which implies that Theorem~\ref{thm3} generalizes Theorem~\ref{thm2}.

The idea is that Theorem~\ref{thm3} provides a method that is robust to constraint violation due to numerical imprecision. Hence, Theorem~\ref{thm3} is directly useful for numerical key rate calculations.  Although it is more complicated than Theorems~\ref{thm1} and \ref{thm2}, Theorem~\ref{thm3} is what we employ in practice for our key rate calculations. For example, the calculations presented in Sec.~\ref{sctexamples} use this theorem.

\subsection{Convergence and tightness}\label{sctTightness}

Theorem~\ref{thm3} is directly useful in key rate calculations, and as such we discuss some of its convergence properties.

In practice, Step 1 does not return a density matrix $\rho^*$ that minimizes $f$ over $\mathbf{S}$. Instead it finds a matrix $\rho_{\epsilon'} \in \St_{\epsilon'}$ that approximately minimizes $f$ over $\St_{\epsilon'}$. Hence, we answer the natural question of how close the lower bound produced from $\rho_{\epsilon'}$ will be to $\alpha = f(\rho^*)$.

Note that by introducing the perturbed function $f_{\epsilon}$, it follows from Lemma~\ref{lemma_gradient_exists} that the lower bound produced by Theorem~\ref{thm3} will be continous with respect to its argument. Thus, if $\rho_{\epsilon\epsilon'}^*$ minimizes $f_{\epsilon}$ over $\St_{\epsilon'}$ and $\rho$ is close to $\rho_{\epsilon\epsilon'}^*$ under some norm, it follows that the lower bounds produced by applying Theorem~\ref{thm3} to $\rho$ and $\rho_{\epsilon\epsilon'}^*$ will be close. Concretely, we have
\begin{align} \label{eqnlimitconvergencerelaxed}
\lim_{\rho \to \rho_{\epsilon \epsilon'}^*} \beta_{\epsilon\epsilon'}(\rho ) - \zeta_\epsilon = \beta_{\epsilon\epsilon'}(\rho_{\epsilon\epsilon'}^*) - \zeta_\epsilon \,,
\end{align}
showing that our lower bounding method converges (i.e. the optimal lower bound of $\rho_{\epsilon\epsilon'}^*$ is approachable). 

In Appendix~\ref{sctAppTightness} we show that the minimizer $\rho^*$ of $f$ over $\mathbf{S}$ and the minimizer $\rho_{\epsilon\epsilon'}^*$ of $f_{\epsilon}$ over $\St_{\epsilon'}$ satisfy
\begin{align} \label{eqnlimittightnessrelaxed}
\lim_{\epsilon,\epsilon' \to 0} \beta_{\epsilon\epsilon'}( \rho_{\epsilon\epsilon'}^* )-\zeta_\epsilon = f(\rho^*) \,.
\end{align}
Therefore, the lower bound produced by Theorem~\ref{thm3} is tight when applied to $ \rho_{\epsilon\epsilon'}^* $.

Combining \eqref{eqnlimitconvergencerelaxed} and \eqref{eqnlimittightnessrelaxed}, it follows that the lower bound produced by applying Theorem~\ref{thm3} to $\rho$ will be arbitrarily close to $f(\rho^*)$ provided that $\epsilon,\epsilon'$ are small and $\rho$ is close to $\rho_{\epsilon\epsilon'}^*$. So provided a suitable computer implementation, our method will produce lower bounds on the true key rate that are arbitrarily tight.

\subsection{Finding a near-optimal attack}\label{sctalgorithm}

Thus far we focused on Step~2 of the key rate calculation procedure. However, Step~1 needs to be addressed since obtaining a reasonable lower bound from Step~2 requires a $\rho \in \mathbf{S}$ that is sufficiently close to the optimal solution $\rho^*$. Here we present an algorithm that has proved effective in practice.

Note that in what follows we do not distinguish between exact and inexact representations of $\{\Gamma_i\}$ or $\{\gamma_i\}$, nor do we distinguish between exact and approximate set membership. The validity of this approach is justified by observing that Step~1 is decoupled from Step~2. Specifically, we can apply Theorem~\ref{thm3} to any positive semidefinite matrix that is approximately consistent with the constraints, and still obtain a reliable lower bound.

By applying the Gram-Schmidt process to the original set of observables $\{\Gamma_i\}$, we obtain a set $\{\bar{\Gamma}_i\}$ of Hermitian operators that are orthogonal under the Hilbert-Schmidt norm. The expectation value of each of these operators is denoted
\begin{align}
\bar{\gamma}_i:= \ave{\bar{\Gamma}_i}  \,,
\end{align}
for states $\rho\in \mathbf{S}$. We extend this set to an orthonormal basis $\{\bar{\Gamma}_i\} \cup \{\Omega_j\}$ for the complete Hermitian operator space with suitable operators $\{\Omega_j\}$. With this basis, we may rewrite \eqref{eqnconstraintSet} as
\begin{align}
\label{eqnSubspaceS}
\mathbf{S} = \left \{ \sum_i \bar{\gamma}_i \bar{\Gamma}_i + \sum_j \omega_j \Omega_j \in \mathbf{H_{+}} \mid \vec{\omega} \in \mathbb{R}^m \right \} \,,
\end{align}
where $m$ is the number of free parameters. This perspective on $\mathbf{S}$ divides the operator space into ``fixed'' and ``free'' subspaces. From a practical optimization point-of-view, the benefit of this representation is that it reduces the original equality constrained problem (which is cumbersome to work with) to a constrained minimization subject to a single semidefinite constraint.

We now adapt the Frank-Wolfe method \cite{Frank1956} to problem \eqref{eqnPAEC4}. As discussed below, the minimization in Algorithm~\ref{alalgorithm1} is a semidefinite program and hence is easier to perform than the original minimization problem in \eqref{eqnPAEC4}.
\begin{algorithm}[H]
\caption{Minimization algorithm for Step 1}\label{alalgorithm1}
\begin{algorithmic}[1]
\State Let $\epsilon > 0$, $\rho_0 \in \mathbf{S}$ and set $i = 0$.
\State Compute $\Delta\rho := \arg \min_{\Delta\rho}\Tr\left[(\Delta\rho)^T  \nabla f(\rho_i)\right]$ subject to $\Delta\rho+\rho_i \in \mathbf{S}$, where ``$\arg \min$'' denotes the argument that achieves the minimization.
\State If $\Tr\left[(\Delta\rho)^T  \nabla f(\rho_i)\right] < \epsilon$ then STOP.
\State Find $\lambda \in (0,1)$ that minimizes $f(\rho_i + \lambda \Delta\rho)$.
\State Set $\rho_{i+1} = \rho_i + \lambda \Delta\rho$, $i \gets i + 1$ and go to 2.
\end{algorithmic}
\end{algorithm}
There are two concrete reasons why adopting the subspace perspective, as in \eqref{eqnSubspaceS}, is useful for solving Algorithm \ref{alalgorithm1}. First, finding a $\rho_0 \in \mathbf{S}$ becomes simple. We only need to find $\vec{\omega}$ so that
\begin{align}
\rho_0 = \sum_i \bar{\gamma}_i \bar{\Gamma}_i + \sum_j \omega_j \Omega_j \in \mathbf{H_{+}}\,.
\end{align}
Second, $\Delta\rho$ has the form
\begin{align}
\Delta\rho = \sum_j \omega_j \Omega_j\,.
\end{align}
Calculating $\vec{\omega} = \{\omega_1, ..., \omega_m \}$ requires solving the standard linear semidefinite program (SDP)
\begin{align}
&\vec{\omega} = \arg \min_{\vec{\omega}} \sum_j \omega_j \Tr\left[\Omega_j^T \nabla f(\rho_i)  \right] \\
&\text{subject to} \sum_j \left( \omega_j \Omega_j \right) + \rho_i \in \mathbf{H_{+}}\,.
\end{align}
Problems of this form have been extensively studied (e.g., see \cite{Boyd2010}) and efficient SDP solvers are widely available.

\section{General framework for protocols}\label{sctqkdframework}

\subsubsection{Prototypical protocol}

Having stated our main result in abstract form, we now connect our result to concrete QKD protocols. Our goal is to present a framework that applies to the known discrete-variable (DV) QKD protocols in the literature. This includes both entanglement-based (EB) and prepare-and-measure (PM) protocols. Examples of protocols that fall under our framework include the BB84 \cite{Bennett1984}, B92 \cite{Bennett1992}, SARG \cite{Scarani2004,Scarani2009}, six-state \cite{Bruss2002}, and decoy-state protocols \cite{Lo2005}. Rather than show how our approach applies to each of these examples, we instead construct a generic protocol that encompasses these examples. We call this the ``prototypical protocol'', shown in Fig.~\ref{fgrPrototypical}. (In Fig.~\ref{fgrPrototypical}, the distinction between the EB and PM scenarios is depicted by a box around the source with a dashed outline, indicating that the source may or may not be located inside Alice's lab.) We will now show how to apply our approach to this prototypical protocol.

We remark that our framework can be easily extended to cover protocols with a central node between Alice and Bob, such as the MDI (measurement-device independent) \cite{Lo2012} and the simplified trusted node \cite{Stacey2015} protocols. We direct the reader to Ref.~\cite{Coles2016} for a discussion of this extension.

Let us now describe the basic steps involved in our prototypical protocol:
\begin{enumerate}
  \item In the EB scenario, Alice and Bob each receive a quantum signal ($A$ and $B$, respectively) from a source and they measure the signal according to the respective POVMs $P^A= \{P^A_j\}$ and $P^B=\{P^B_j\}$, producing the raw data. In the PM scenario, the same mathematical description applies (via the so-called source-replacement scheme \cite{Bennett1992a,Ferenczi2012}) since one can think of Alice's prepared states as resulting from Alice performing a measurement $P^A$ on a register system $A$. 
  \item Alice and Bob make a public announcement, announcing some aspect of their measurement outcomes.
  \item Alice and Bob perform post-selection based on these announcements.
  \item Alice implements a key map. The key map is a function that maps Alice's raw data and the announcements to a key symbol, chosen from $\{0,1,...,  N-1\}$ where $N$ is the number of key symbols.
  \item Alice performs one-way error correction, leaking some information to Bob, and Bob forms his key.
  \item Alice performs privacy amplification, typically by applying a random universal hash function and then communicating the choice of hash function to Bob.
\end{enumerate}

\begin{figure}
\begin{center}
\includegraphics[width=3.3in]{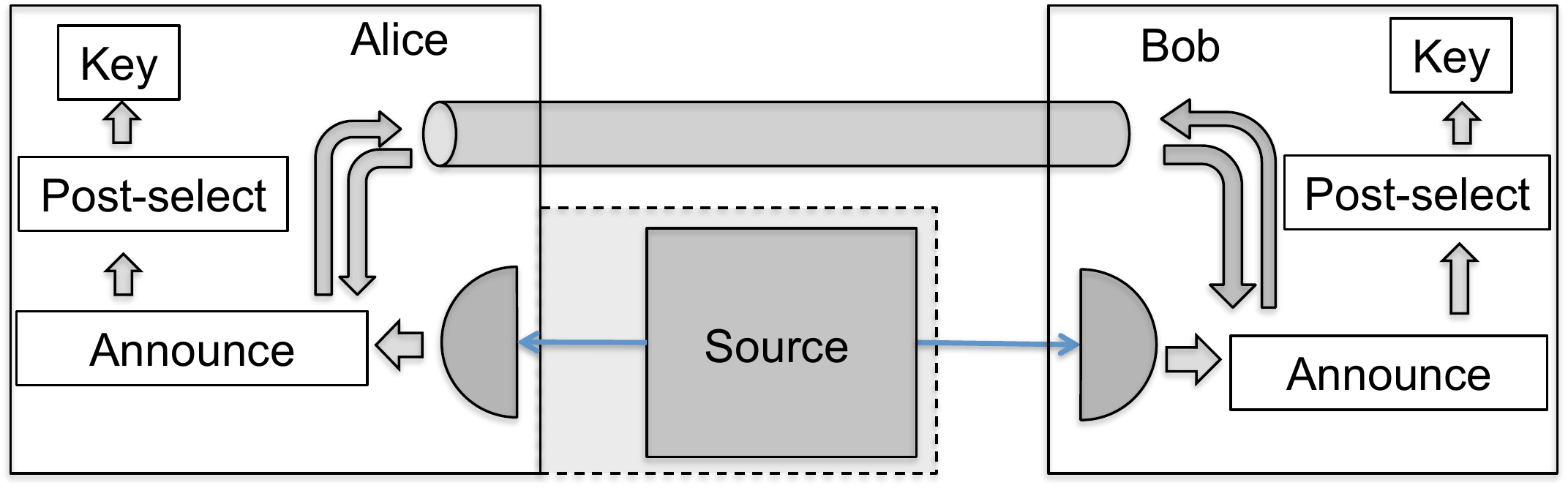}

\caption{A prototypical QKD protocol that we use to illustrate our framework. The source sends systems $A$ and $B$ to Alice and Bob, respectively. The dashed line around the source indicates that it may or may not be inside of Alice's laboratory, respectively treating the PM and EB scenarios. Alice and Bob respectively measure POVMs $P^A$ and $P^B$. Then they publicly announce some information related to their measurement outcomes. These announcements are used to perform post-selection, where some announcement outcomes are discarded. Alice then implements a key map that maps her information (raw data + announcements) to a key variable. Finally, Alice leaks some information (about the result of the key map) to Bob for error correction purposes, and then Bob forms his key.}
\label{fgrPrototypical}
\end{center}
\end{figure}

\subsubsection{Mathematical model for prototypical protocol}

The quantum state shared by Alice and Bob (prior to their measurements) can be written as $\rho_{AB}$. In the worst-case scenario, Eve possesses a purification of $\rho_{AB}$ denoted as system $E$, and we assume this worst-case scenario to be as pessimistic as possible. In the following discussion of the protocol, we will note that Eve obtains access to additional systems due to public announcements made by Alice and Bob. (In total, by the end of the protocol, Eve will have access to $E\At \Bt$ where $\At$ and $\Bt$ are respectively the registers that store Alice's and Bob's public announcements.)

Consider the experimental constraints on the state $\rho_{AB}$. These constraints have the form:
\begin{align} \label{eqnconstraints}
& \Tr((P^A_j \ot P^B_k ) \rho_{AB}) = p_{jk}\,.
\end{align}
For prepare-and-measure (PM) protocols, we add additional constraints, as follows. We employ the source-replacement scheme \cite{Bennett1992a,Ferenczi2012}, which treats system $A$ as a register that stores the information about which state Alice prepared. This corresponds to Alice preparing the bipartite state
\begin{align} \label{eqnSourceReplace94}
\ket{\psi}_{AA'} = \sum_i \sqrt{p_i}\ket{i}_A \ket{\phi_i}_{A'}
\end{align}
where $\{\ket{\phi_i}\}$ are the signal states and $\{p_i\}$ are their associated probabilities. Eve's attack maps system $A'$ to system $B$, producing the state $\rho_{AB}$. On the other hand, system $A$ is inaccessible to Eve, and hence 
\begin{align} \label{eqnSourceReplace95}
\rho_A = \Tr_{B}(\rho_{AB})  = \Tr_{A'}(\dya{\psi}_{AA'})
\end{align}
is fixed, independent of Eve's attack. To fix $\rho_A$ we add in constraints of the form
\begin{align} \label{eqnconstraints22}
\Tr( ( \Theta_j \ot \id_{B}) \rho_{AB} ) = \theta_j
\end{align}
where $\{\Theta_j\}$ is a finite set of tomographically complete observables on $A$. Hence, \eqref{eqnconstraints} and \eqref{eqnconstraints22} together represent the constraints that Alice and Bob have on their state.

Step 2 of the above protocol involves Alice and Bob each making an announcement based on their measurement results. In this case, it is helpful to group together the POVM elements into sets associated with particular announcements. Let us write Alice's POVM as $P^A = \{P^A_j\} =\{P^A_{(a,\alpha_a)}\}$ and Bob's POVM as $P^B = \{P^B_k\} =\{P^B_{(b,\beta_b)}\}$. Here the first index denotes the announcement, and the second index denotes a particular element associated with that announcement. Bob's announcement is given by a quantum channel with Kraus operators 
\begin{align} \label{eqnkrausbob}
K^B_b &= \sum_{\beta_b} \sqrt{  P^B_{(b,\beta_b)} } \ot \ket{b}_{\Bt} \ot \ket{\beta_b}_{\Bb} \,.
\end{align}
We remark that $K^B_b$ expands the Hilbert space by introducing the registers $\Bt$ and $\Bb$, which is why the operators on these subsystems appear as kets in \eqref{eqnkrausbob}. Similarly, Alice's announcement is given by a quantum channel with Kraus operators
\begin{align} \label{eqnkrausalice}
K^A_a &=   \sum_{\alpha_a} \sqrt{ P^A_{(a, \alpha_a)} } \ot \ket{a}_{\At} \ot \ket{\alpha_a}_{\Ab}\,.
\end{align}
Here, $\At$ and $\Bt$ are registers that store Alice's and Bob's announcements, respectively. Also, $\Ab$ and $\Bb$ are registers that store Alice's and Bob's measurement outcomes, for a given announcement. So, the state after making these announcements becomes
\begin{align} \label{eqnstateafterannounce}
  \rho^{(2)}_{A\At \Ab B \Bt \Bb }&=\AC(\rho_{AB})\\
 &= \sum_{a,b}(K^A_a\ot K^B_b)\rho_{AB}(K^A_a\ot K^B_b)\ad\,,
\end{align}
where $\AC$ is a CPTP map. We remark that the form of \eqref{eqnstateafterannounce} is such that $\At$ and $\Bt$ are classical registers, meaning that the purifying system has a copy of the registers. This is the way one models a public announcement.

Step 3 is post-selection. Here, Alice and Bob select some announcements to keep and some to discard. Let $\mathbf{A}$ be the set of all announcements that are kept. Then define the projector
\begin{align} \label{eqnstateafterannounce22}
 \Pi =  \sum_{(a,b)\in \mathbf{A}} \dya{a}_{\At} \ot \dya{b}_{\Bt}   \,,
\end{align}
with identity acting on the other systems. The post-selection is modeled by projecting with this projector to obtain the state
\begin{align} \label{eqnstateafterannounce33}
 \rho^{(3)}_{A\At \Ab B \Bt \Bb } = \frac{\Pi\rho^{(2)}_{A\At \Ab B \Bt \Bb   } \Pi }{p_{\pass}}  \,.
\end{align}
where $p_{\pass} = \Tr(\rhot_{A\At \Ab B \Bt \Bb } \Pi)$.

In step 4, Alice chooses a key map. A key map is a function $g$ whose arguments include the outcome of Alice's measurements $(a,\alpha_a)$ and Bob's announcement $b$. The function outputs a value in $\{0,1,..., N-1\}$ where $N$ is the number of key symbols. Hence, we write the key map as the function $g(a,\alpha_a, b)$. We define an isometry $V$ that stores the key information in a register system $R$, as follows
\begin{align} \label{eqnstateafterannounce44}
V = \sum_{a,\alpha_a, b} &\ket{g(a,\alpha_a, b)}_R \ot \dya{a}_{\At} \ot \dya{\alpha_a}_{\Ab} \ot \dya{b}_{\Bt}\,.
\end{align}
We first act with this isometry on the state $\rho^{(3)}$ in \eqref{eqnstateafterannounce33}, 
\begin{align}
 \rho^{(4)}_{R A\At \Ab B \Bt \Bb } = V \rho^{(3)}_{A\At \Ab B \Bt \Bb } V\ad   \,,
\end{align}
which stores the key information in the standard basis $\{\ket{j}_R\}$ on $R$. Then we decohere $R$ in this basis, which turns $R$ into a classical register denoted $Z^R$, giving the final state
\begin{align} \label{eqnstateafterannounce66}
 \rho^{(5)}_{Z^R A\At \Ab B \Bt \Bb } = \ZC\left(\rho^{(4)}_{R A\At \Ab B \Bt \Bb }  \right)  \,,
\end{align}
where $\ZC$ is a pinching quantum channel, whose action is given by $\ZC(\sigma) = \sum_j (\dya{j}_R \ot \id) \sigma (\dya{j}_R \ot \id)$.

\subsubsection{Key rate}

Finally, Alice performs error correction, which gives $\leak$ number of bits about the key map results to Eve, followed by privacy amplification.  This gives the following formula for the key rate \cite{Devetak2005}:
\begin{align}
\label{eqnkeyrate1}
K &= p_{\pass}\left[ H(Z^R | E \At \Bt   )_{\rho^{(5)}} - \leak \right] \,.
\end{align} 
Here, $H(A|B)_{\rho} = H(\rho_{AB}) - H(\rho_B)$ denotes the conditional von Neumann entropy with $H(\sigma) = -\Tr(\sigma \log \sigma)$. Note that the first term in \eqref{eqnkeyrate1} refers to the state $\rho^{(5)}_{Z^R A\At \Ab B \Bt \Bb }$ defined in \eqref{eqnstateafterannounce66}. Also, as noted above, $E$ is a purifying system of $\rho_{AB}$.

More precisely, the expression in \eqref{eqnkeyrate1} must be minimized over all density operators $\rho_{AB}$ that satisfy the constraints in \eqref{eqnconstraints} and \eqref{eqnconstraints22}. Hence we write:
\begin{align}
\label{eqnkeyrate33}
K &=  \min_{\rho_{AB}\in \mathbf{S} }\left( p_{\pass} H(Z^R | E \At \Bt   )_{\rho^{(5)}}\right) - p_{\pass} \leak \,,
\end{align} 
where $\mathbf{S}$ has the general form in \eqref{eqnconstraintSet}, with the constraints given by \eqref{eqnconstraints} and \eqref{eqnconstraints22}. 

Having now expressed the key rate in explicit form, we can now relate \eqref{eqnkeyrate33} back to the discussion in Sec.~\ref{sctbackground}. This involves simplifying the notation. Namely, for the constraints in \eqref{eqnconstraints} and \eqref{eqnconstraints22}, we rewrite them as $\Tr(\Gamma_i \rho) = \gamma_i$, as in \eqref{eqnconstraints1234}. Note that we drop the subsystem labels on the state $\rho = \rho_{AB}$. Finally, we write the optimization problem in \eqref{eqnkeyrate33} as
\begin{align}
\alpha &= \min_{\rho \in \mathbf{S}}f(\rho)\label{eqnprimal16342}\,,
\end{align}
where
\begin{align}
\label{eqnfrho1}
f(\rho) &= p_{\pass} \cdot H(Z^R | E \At \Bt   )_{\rho^{(5)}} \\
\label{eqnfrho2}
&= p_{\pass} \cdot D\left(\rho^{(4)}_{R A\At \Ab B \Bt \Bb} || \rho^{(5)}_{Z^R A\At \Ab B \Bt \Bb} \right)\\
\label{eqnfrho232}
&= D\left(\GC(\rho_{AB}) || \ZC(\GC(\rho_{AB} )) \right)\,.
\end{align}
Note that \eqref{eqnfrho2} removes the dependence on Eve's system $E$ and is derived using Theorem 1 from \cite{Coles2012}. Equation \eqref{eqnfrho232} is derived from the previous line using the property $D( c \sigma || c \tau) = c D( \sigma || \tau)$ for any constant $c > 0$.  Furthermore we define $\GC$ such that its action on an operator $\sigma$ is given by
\begin{align}
\label{eqnfrho235431}
\GC(\sigma) =   V \hspace{2pt} \Pi  \hspace{2pt}\AC(\sigma)  \hspace{2pt} \Pi  \hspace{2pt}V\ad\,.
\end{align}

Note that \eqref{eqnfrho232} has the same form as Eq.~\eqref{eqnPAEC4b}. In summary, to apply our numerical approach to a given protocol, one formulates $\GC$ via \eqref{eqnfrho235431} and the constraints via \eqref{eqnconstraints} and \eqref{eqnconstraints22}. With these objects defined, one then applies our numerical method outlined in Sec.~\ref{sctmainresult}.

\section{Examples}\label{sctexamples}

In this section we consider three practically important scenarios that show the power of our approach. In particular, we consider (1) the BB84 protocol with detector efficiency mismatch, (2) the Trojan-horse attack on the BB84 protocol, and (3) the BB84 protocol with phase-coherent signal states. Each of these three scenarios involves a BB84-style protocol but with some ``imperfection'' accounted for (detector inefficiency, existence of a side channel, and lack of phase randomization). For each scenario, our approach yields significantly higher key rates than those previously obtained in the literature. We remark that the robustness of our numerical approach could allow us to investigate all three imperfections in a single protocol, although for simplicity we consider them separately.

For illustration purposes, for the following examples we assume that error correction is performed at the Shannon limit. This means that the error correction term in \eqref{eqnkeyrate33} can be written as a conditional entropy,
\begin{align}
\leak &= H(Z^R | Z^{\Bb} \At \Bt  )_{\rho^{(5)}}\,,
\end{align}
with the state $\rho^{(5)}$ defined in \eqref{eqnstateafterannounce66}. Here, $Z^{\Bb}$ can be viewed as the classical register that one would obtain from measuring in the standard basis on $\Bb$. 

Let us briefly remark about the gaps between our upper and lower bounds on the key rate, where the upper and lower bounds are respectively obtained from our "Step 1" and "Step 2" from Sec.~\ref{sctmainresult}. In all of our figures below, we only plot our lower bounds, since they are guaranteed to be reliable. But it is worth noting that, for the examples considered below, the gaps between our upper and lower bounds are extremely small. Namely we observed the gap to be $\sim 10^{-6}$ times smaller than the key rate value. Hence for these examples our upper and lower bounds essentially coincide. In these cases, one can think of our Step 2 as simply verifying that Step 1 worked well, i.e., that Step 1 found the optimal attack. We remark that the examples considered below are low dimensional, and it is possible that the gap between the upper and lower bounds could grow with the number of signal states, which could cause Step 1 to be less reliable.

For each example, further details about the constraints used in our calculations can be found in Appendix~\ref{sctAppExamples}.

\subsection{Efficiency mismatch}\label{sctEffMis}

Consider a polarization-encoded BB84 protocol, where Bob actively choses his detection basis setting. In this case, Bob's measurement involves two detectors, $D_1$ and $D_2$, that are associated with the two polarization states for a given basis. 

In practice, it is likely that $D_1$ and $D_2$ have different efficiencies, commonly referred to as efficiency mismatch. This mismatch can be further enhanced (by Eve) by manipulating the spatial mode of the incoming light~\cite{Sajeed2015}. When efficiency mismatch is large enough, successful hacking strategies on QKD systems have been demonstrated~\cite{Lydersen2010a}. Furthermore, even with a small amount of efficiency mismatch, the security analysis of QKD becomes difficult to perform. 

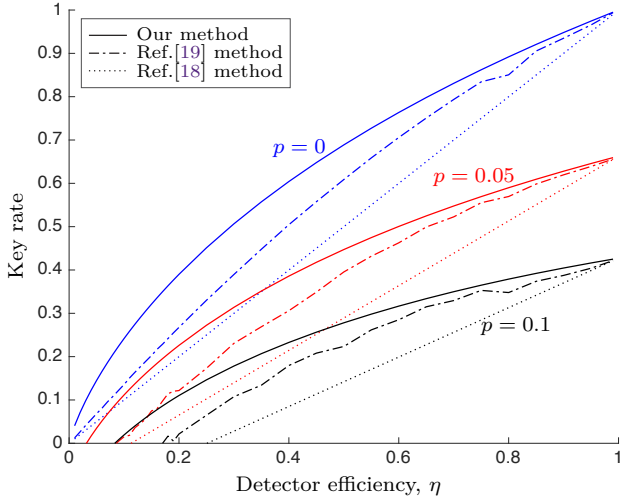
\begin{figure}[tbp]
\begin{center}
\vspace{0.2cm}
\hspace{-0.01cm}
\begin{overpic}[width=7.9cm]{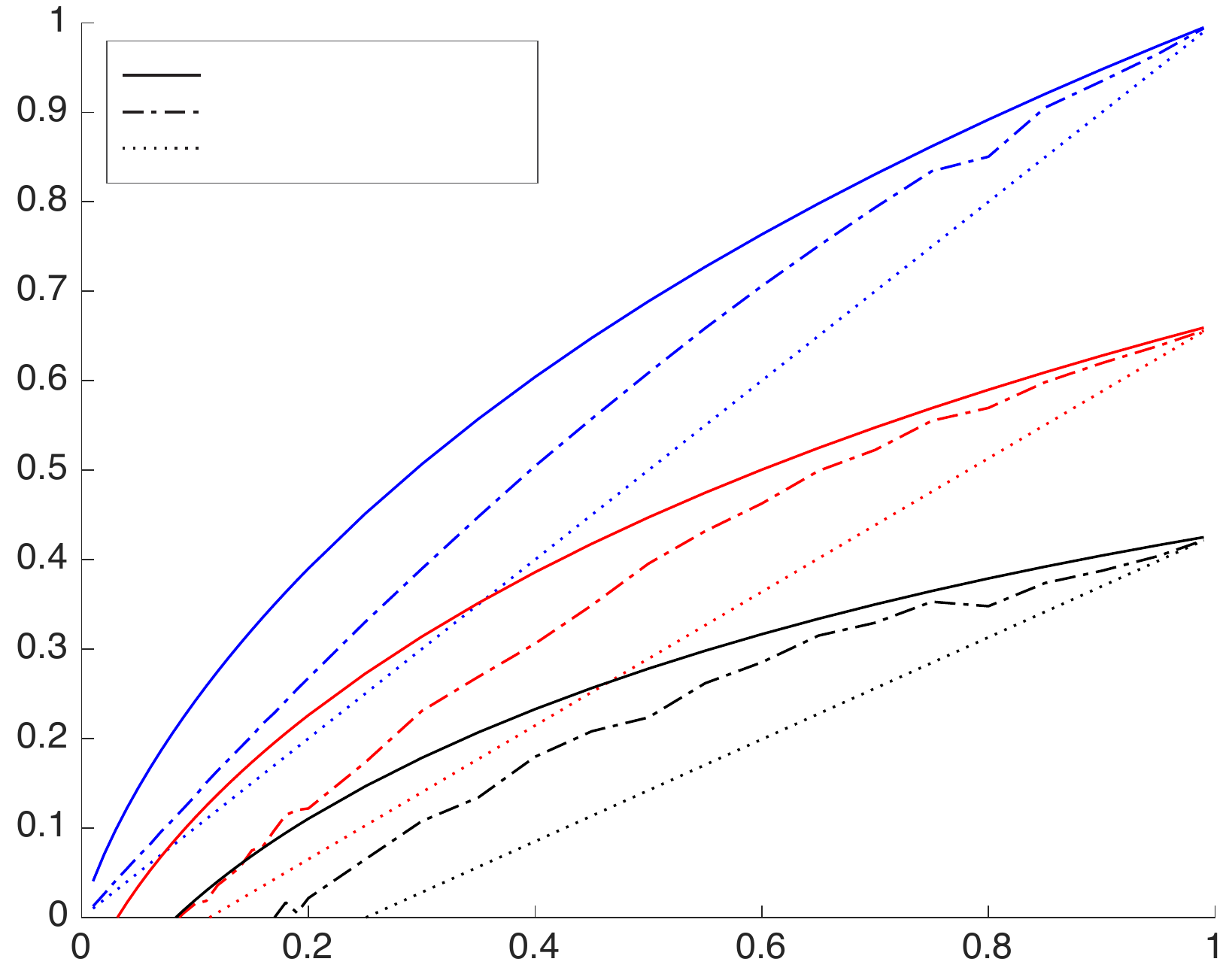}

\put(18,73){\scriptsize Our method}
\put(18,69.9){\scriptsize Ref.\cite{Coles2016} method}
\put(18,66.8){\scriptsize Ref.\cite{Fung2009} method}

\put(41,54){\color{blue}\footnotesize $p = 0$}
\put(68,49){\color{red}\footnotesize $p = 0.05$}
\put(76,24){\color{black}\footnotesize $p = 0.1$}

\put(52,-2){\makebox(0,0){\footnotesize Detector efficiency, $\eta$}}
\put(-2,40){\makebox(0,0){\footnotesize \rotatebox{90}{Key rate}}}
\end{overpic}
\vspace{0.2cm}
\caption{Key rate for the BB84 protocol with detector efficiency mismatch. Curves are shown for three values of depolarizing probability $p$ (0, 0.05, 0.1). The $x$-axis is the efficiency of the least efficient detector, with the other detector's efficiency being set to one.}
\label{fgrEffMismatch}
\end{center}
\end{figure}

This motivates the application of our numerical approach the case of detector efficiency mismatch. For simplicity, we consider single-photon signal states, and we assume that no multiple photons arrive at Bob's measurement apparatus. (Our numerical approach can handle multi-photon signals and multi-photon detection events; however, we leave a detailed discussion of the general case for future work.) The single-photon case was previously treated analytically by Fung et al.~\cite{Fung2009}, and hence it provides an opportunity to compare our numerics to the literature.

To further simplify the analysis, we assume one of Bob's detectors is perfect while the other detector has an efficiency $\eta$. This allows us to plot the key rate as a function of $\eta$, as shown in Fig.~\ref{fgrEffMismatch}. The curves are shown for various depolarizing noise levels $p$, where 
\begin{align}
\rho \rightarrow (1-p) \rho + p \id/d
\end{align}
is the depolarizing channel. (We emphasize that our security analysis does not assume that the depolarizing channel is the actual attack; we simply use this channel to generate artificial data for Alice and Bob.) The plot shows that our new numerical method outperforms our previous numerical method based on the dual problem \cite{Coles2016}, which in turn outperforms the analytical method from Ref.~\cite{Fung2009}.

\subsection{Trojan-horse attack}\label{sctTHA}

Consider the phase-encoded BB84 protocol. Here, Alice's light source produces a pulse that passes through an interferometer, one arm of which applies a variable phase $\theta$ chosen from the set $\{0,\pi / 2 , \pi, 3\pi / 2\}$ to encode the information. Bob decodes this phase information with an interferometer in his lab.

\begin{figure}[tbp]
\begin{center}
\vspace{0.2cm}
\hspace{-0.01cm}
\begin{overpic}[width=7.9cm]{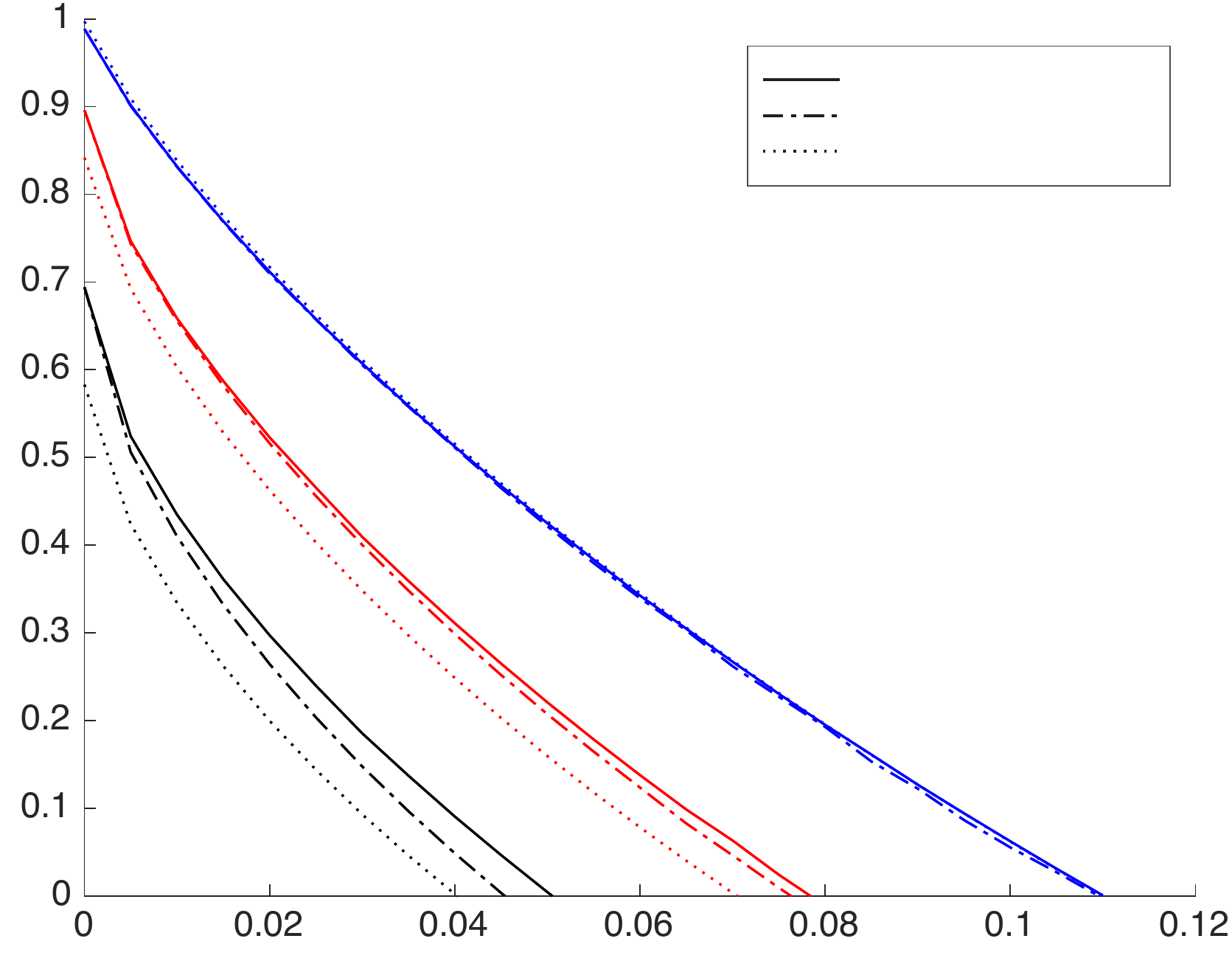}

\put(69.5,70.7){\scriptsize Our method}
\put(69.5,67.6){\scriptsize Ref.\cite{Coles2016} method}
\put(69.5,64.5){\scriptsize Ref.\cite{Lucamarini2015} method}

\put(52,32){\color{blue}\footnotesize $\mu_\text{out} = 0$}
\put(34,22){\color{red}\footnotesize $\mu_\text{out} = 0.01$}
\put(12,10){\color{black}\footnotesize $\mu_\text{out} = 0.04$}

\put(52,-2){\makebox(0,0){\footnotesize Error rate, $Q$}}
\put(-2,40){\makebox(0,0){\footnotesize \rotatebox{90}{Key rate}}}
\end{overpic}
\vspace{0.2cm}

\caption{Key rate vs error rate for the single-photon BB84 protocol under a Trojan-horse attack. The key rate is plotted for different values of $\mu_\text{out}$. Our numerical method improves on our previous approach in Ref.~\cite{Coles2016}, which in turn gives higher key rates than the analytical method of Ref.~\cite{Lucamarini2015}.}
\label{fgrTrojan}
\end{center}
\end{figure}
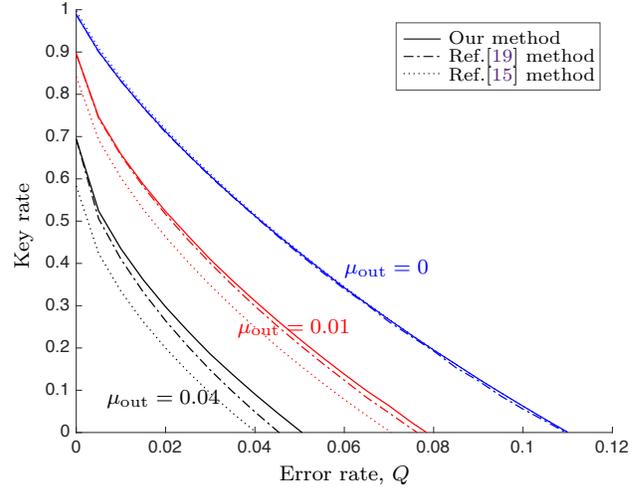

There is a simple hacking attack on this protocol that exploits a side channel in Alice's encoder (i.e., a channel by which Eve can obtain additional information, beyond the direct channel from Alice to Bob). The attack involves Eve sending a bright pulse of light into Alice's lab \cite{Vakhitov2001}. Some fraction of this pulse reaches Alice's phase encoder and is encoded with the same information that Alice is attempting to send to Bob. A portion of this light is reflected back to Eve, who can then decode some of the phase information, potentially compromising the protocol's security. This is called the Trojan-horse attack (THA) \cite{Gisin2006}, since it involves a ``malicious gift'' from Eve.

For simplicity, let us restrict our attention here to the case where Alice's light source outputs only a single photon per signal. Following the approach of Lucamarini et al.~\cite{Lucamarini2015}, we describe Eve's input pulse and back-reflected pulse as coherent states denoted by $\ket{\sqrt{\mu_\text{in}}}$ and $\ket{e^{i\theta}\sqrt{\mu_\text{out}}}$, respectively. Here, $\theta$ stores Alice's phase encoding setting, and the input and output intensities typically satisfy $\mu_\text{out} \ll \mu_\text{in}$. The signal states emerging from Alice's source are:
\begin{align}
\ket{\phi_{z+}} &= \ket{z_{+}}_{S} \ot \ket{+\sqrt{\mu_\text{out}}}_{S'} \label{eqntrojansignal1} \\
\ket{\phi_{z-}} &= \ket{z_{-}}_{S} \ot \ket{-\sqrt{\mu_\text{out}}}_{S'} \label{eqntrojansignal2} \\
\ket{\phi_{x+}} &= \ket{x_{+}}_{S} \ot \ket{+i\sqrt{\mu_\text{out}}}_{S'} \label{eqntrojansignal3} \\
\ket{\phi_{x-}} &= \ket{x_{-}}_{S}\ot \ket{-i\sqrt{\mu_\text{out}}}_{S'}\,, \label{eqntrojansignal4}
\end{align}
where
\begin{align}
\ket{z_{\pm}} &:= \frac{1}{\sqrt{2}}(\ket{1}_L \ket{0}_M \pm \ket{0}_L \ket{1}_M)\\
\ket{x_{\pm}} &:= \frac{1}{\sqrt{2}}(\ket{1}_L \ket{0}_M \pm i \ket{0}_L \ket{1}_M)\,.
\end{align}
Here, $\ket{n}_L$ ($\ket{n}_M$) is the $n$-photon state of the long (short) arm of the interferometer.

Figure~\ref{fgrTrojan} shows the results of our numerics for the THA. Here we plot key rate versus error rate (assuming the $z$ and $x$ error rates are identical, for simplicity) for various values of $\mu_{\text{out}}$. The plot shows that our new numerical method gives higher key rates than both the analytical method from Ref.~\cite{Lucamarini2015} as well as the numerical method from Ref.~\cite{Coles2016}.

\subsection{BB84 protocol with phase-coherent signal states}\label{sctLoPreskill}

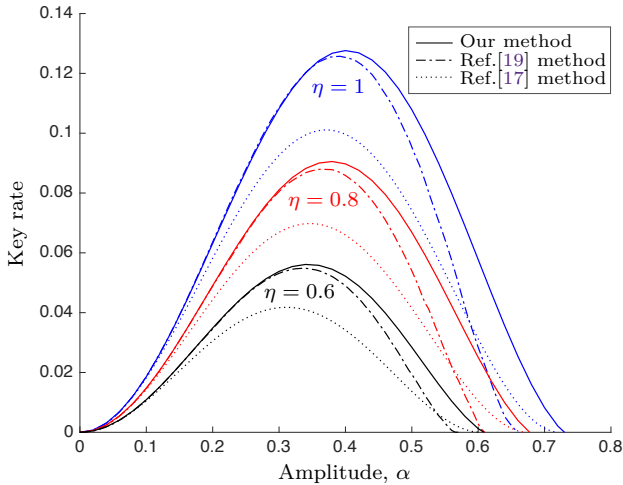
\begin{figure}[t!]
\begin{center}
\vspace{0.2cm}
\hspace{-0.01cm}
\begin{overpic}[width=7.9cm]{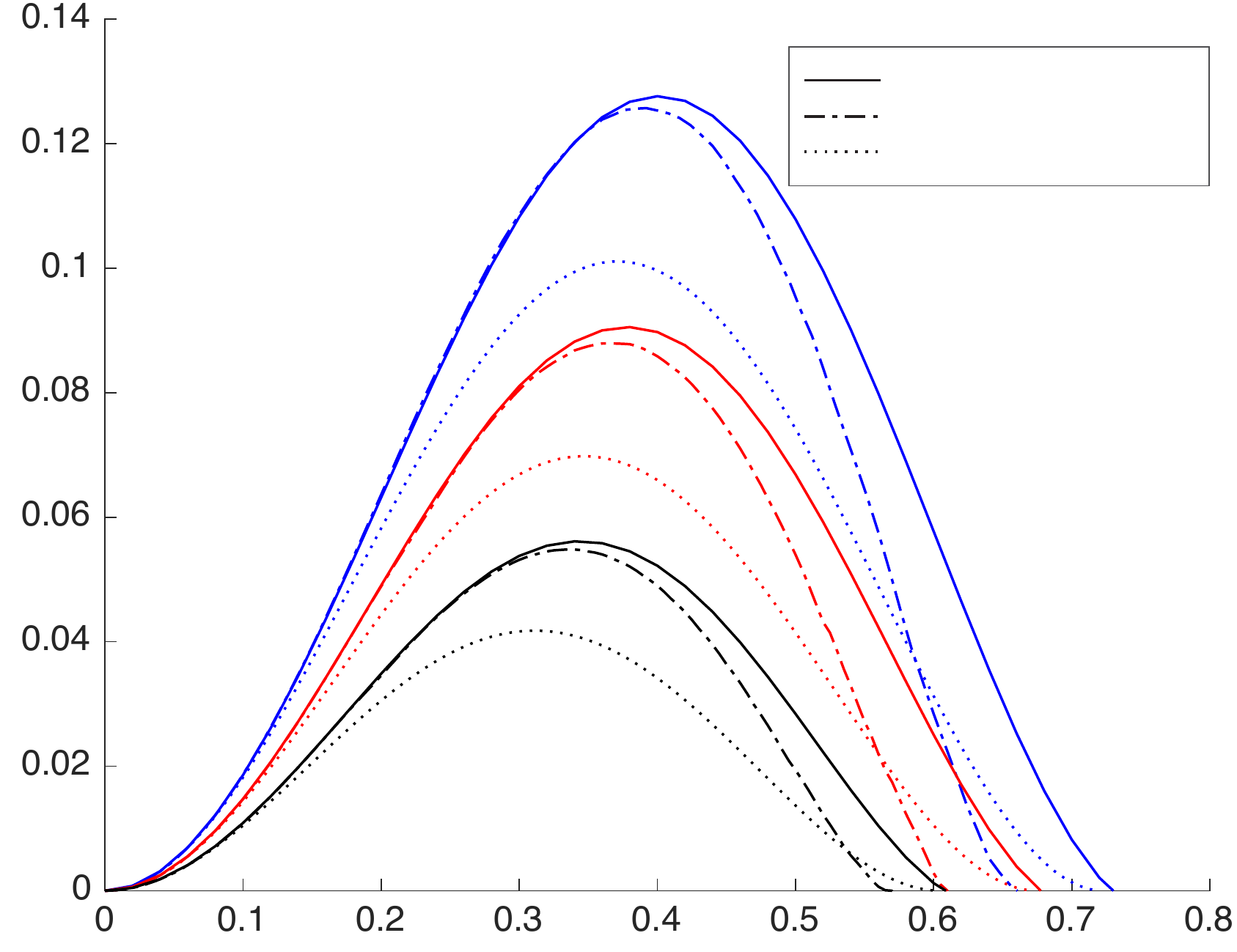}

\put(72.4,69.8){\scriptsize Our method}
\put(72.4,66.7){\scriptsize Ref.\cite{Coles2016} method}
\put(72.4,63.6){\scriptsize Ref.\cite{Lo2007} method}

\put(47.5,62){\color{blue}\footnotesize $\eta = 1$}
\put(43.5,43){\color{red}\footnotesize $\eta = 0.8$}
\put(39.5,27.5){\color{black}\footnotesize $\eta = 0.6$}

\put(53,-2){\makebox(0,0){\footnotesize Amplitude, $\alpha$}}
\put(-2,40){\makebox(0,0){\footnotesize \rotatebox{90}{Key rate}}}
\end{overpic}
\vspace{0.2cm}
\caption{Key rate versus signal-state amplitude $\alpha$ for three values of transmission probability - $\eta = 1,0.8,0.6$ - for the Huttner et al.~\cite{Huttner1995} protocol.}
\label{fgrLP}
\end{center}
\end{figure}

Here we consider a protocol proposed by Huttner et al.~\cite{Huttner1995} and analyzed by Lo and Preskill~\cite{Lo2007}. This is a phase-encoded BB84 protocol, but using coherent states instead of single-photon states. This is quite practical, since one can use attenuated laser pulses from mode-locked lasers to generate these signal states. In addition the protocol is practical because the experimenter does not need to do phase randomization. So it is worth investigating the key rate of this protocol.

The signal states prepared by Alice are~\cite{Lo2007}
\begin{align}
\ket{\phi_{z+}} &= \ket{+\alpha}_S \otimes \ket{\alpha}_{S'} \label{eqnlpsignal1} \\
\ket{\phi_{z-}} &= \ket{-\alpha}_S \otimes \ket{\alpha}_{S'} \label{eqnlpsignal2} \\
\ket{\phi_{x+}} &= \ket{+i\alpha}_S \otimes \ket{\alpha}_{S'} \label{eqnlpsignal3} \\
\ket{\phi_{x-}} &= \ket{-i\alpha}_S \otimes \ket{\alpha}_{S'} \label{eqnlpsignal4}
\end{align}
where $\alpha$ is the amplitude of the coherent state, $S$ is the signal mode and $S'$ is the reference mode. When Bob receives the signal, he performs a polarization measurement (in one of two complementary bases), discarding no-click events, and assigning a random bit value to double-click events.

Lo and Preskill~\cite{Lo2007} gave an analytical lower bound on the key rate for this protocol, as a function of transmission probability $\eta$ and amplitude $\alpha$. Their theoretical curves are shown as dotted lines in Fig.~\ref{fgrLP}, for several values of $\eta$. In the same plot, we show the result of our numerical optimization as solid lines, with the key rates obtained from the method in Ref.~\cite{Coles2016} shown as dashed-dotted lines. Interestingly our numerics give higher key rates than the previous literature over the entire parameter range. This is an important result due to the practicality of this protocol.

\section{Conclusions}\label{sctconclusion}

In conclusion, we presented a new numerical approach for calculating key rates for QKD. For concreteness, we name our approach the ``reliable primal method'' or the ``reliable primal problem''. As discussed in Sec.~\ref{sctmainresult}, Step~1 of our method is simply the primal optimization problem. Step~2 of our method converts the output of Step~1 (a nearly optimal eavesdropping attack) into a reliable lower bound on the key rate. We presented an efficient method for Step~1 in Sec.~\ref{sctalgorithm}. Our main contribution is a method for Step~2, which is presented in Theorems~\ref{thm1}, \ref{thm2}, and \ref{thm3}.

Reliability is the most important issue with numerical key rate calculations, since key rates must come with a security guarantee. In this work, we highlighted the various issues associated with numerical key rate calculations, such as constraint violation and inexact variable storage by computers. Furthermore, we showed how to address these issues. Our most robust result, Theorem~\ref{thm3}, allows one to lower bound the key rate despite numerical imprecision.  

We discussed that our method is arbitrarily tight in Sec~\ref{sctTightness}. This allowed us to make significant improvements over previous literature key rates for three interesting examples in Sec.~\ref{sctexamples}. Furthermore, the tightness of our approach implies that the solid curves that we plotted in Figs.~\ref{fgrEffMismatch}, \ref{fgrTrojan}, and \ref{fgrLP} are essentially unbeatable, i.e., they cannot be improved upon. Eliminating looseness from key rate calculations is a major advance for the field of QKD research. 

Future applications of our work include investigating device imperfections, side channels, multi-photon detection events, decoy-state protocols with partial phase randomization, measurement-device independent protocols, differential-phase shift protocols, and coherent one-way protocols. Perhaps more importantly, our approach can allow researchers to explore and evaluate novel protocol ideas that have yet to be discovered.

As noted in the Introduction, our group released a user-friendly software for key rate calculations based on the dual problem from Ref.~\cite{Coles2016}. Interestingly, our reliable primal method presented here improves on the approach of Ref.~\cite{Coles2016} in terms of both speed and tightness. Therefore, we plan to improve our publicly-available software in the future by incorporating the reliable primal method. We believe this software has the potential to be used throughout the QKD community, both in industry and academia.

Finally, we hope to extend our approach to finite-key analysis \cite{Scarani2008, Renner2005, Sano2010, Tomamichel2012a} in the near future.

\section{Acknowledgements}\label{sctackknowledge}

We thank Yanbao Zhang for preparing and doing the calculations for Fig.~\ref{fgrEffMismatch}. We are furthermore grateful for helpful discussions with Dr. Zhang and Shaun Ren. We acknowledge support from Industry Canada, Sandia National Laboratories, Office of Naval Research (ONR), NSERC Discovery Grant, and Ontario Research Fund (ORF).

\bibliographystyle{naturemag}
\bibliography{primal3}

\begin{thebibliography}{10}
\expandafter\ifx\csname url\endcsname\relax
  \def\url#1{\texttt{#1}}\fi
\expandafter\ifx\csname urlprefix\endcsname\relax\def\urlprefix{URL }\fi
\providecommand{\bibinfo}[2]{#2}
\providecommand{\eprint}[2][]{\url{#2}}

\bibitem{Campagna2015}
\bibinfo{author}{Campagna, M.} \emph{et~al.}
\newblock \emph{\bibinfo{title}{{Quantum Safe Cryptography and Security}}}
  (\bibinfo{publisher}{European Telecommunications Standards Institute},
  \bibinfo{year}{2015}).

\bibitem{Wyner1975}
\bibinfo{author}{Wyner, A.~D.}
\newblock \bibinfo{title}{{The Wire-Tap Channel}}.
\newblock \emph{\bibinfo{journal}{Bell System Technical Journal}}
  \textbf{\bibinfo{volume}{54}}, \bibinfo{pages}{1355--1387}
  (\bibinfo{year}{1975}).
\newblock \urlprefix\url{https://doi.org/10.1002/j.1538-7305.1975.tb02040.x}.

\bibitem{Scarani2009}
\bibinfo{author}{Scarani, V.} \emph{et~al.}
\newblock \bibinfo{title}{{The security of practical quantum key
  distribution}}.
\newblock \emph{\bibinfo{journal}{Reviews of Modern Physics}}
  \textbf{\bibinfo{volume}{81}}, \bibinfo{pages}{1301--1350}
  (\bibinfo{year}{2009}).
\newblock \urlprefix\url{https://doi.org/10.1103/RevModPhys.81.1301}.

\bibitem{Lo2014}
\bibinfo{author}{Lo, H.-K.}, \bibinfo{author}{Curty, M.} \&
  \bibinfo{author}{Tamaki, K.}
\newblock \bibinfo{title}{{Secure quantum key distribution}}.
\newblock \emph{\bibinfo{journal}{Nature Photonics}}
  \textbf{\bibinfo{volume}{8}}, \bibinfo{pages}{595--604}
  (\bibinfo{year}{2014}).
\newblock \urlprefix\url{https://doi.org/10.1038/nphoton.2014.149}.

\bibitem{Xin2011}
\bibinfo{author}{Xin, H.}
\newblock \bibinfo{title}{{Chinese Academy Takes Space Under Its Wing}}.
\newblock \emph{\bibinfo{journal}{Science}} \textbf{\bibinfo{volume}{332}},
  \bibinfo{pages}{904--904} (\bibinfo{year}{2011}).
\newblock \urlprefix\url{https://doi.org/10.1126/science.332.6032.904}.

\bibitem{Peev2009}
\bibinfo{author}{Peev, M.} \emph{et~al.}
\newblock \bibinfo{title}{{The SECOQC quantum key distribution network in
  Vienna}}.
\newblock \emph{\bibinfo{journal}{New Journal of Physics}}
  \textbf{\bibinfo{volume}{11}}, \bibinfo{pages}{075001}
  (\bibinfo{year}{2009}).
\newblock \urlprefix\url{https://doi.org/10.1088/1367-2630/11/7/075001}.

\bibitem{Sasaki2011}
\bibinfo{author}{Sasaki, M.} \emph{et~al.}
\newblock \bibinfo{title}{{Field test of quantum key distribution in the Tokyo
  QKD Network.}}
\newblock \emph{\bibinfo{journal}{Optics Express}}
  \textbf{\bibinfo{volume}{19}}, \bibinfo{pages}{10387--10409}
  (\bibinfo{year}{2011}).
\newblock \urlprefix\url{https://doi.org/10.1364/OE.19.010387}.

\bibitem{Wang2014}
\bibinfo{author}{Wang, S.} \emph{et~al.}
\newblock \bibinfo{title}{{Field and long-term demonstration of a wide area
  quantum key distribution network}}.
\newblock \emph{\bibinfo{journal}{Optics Express}}
  \textbf{\bibinfo{volume}{22}}, \bibinfo{pages}{21739} (\bibinfo{year}{2014}).
\newblock \urlprefix\url{https://doi.org/10.1364/OE.22.021739}.

\bibitem{Renner2005}
\bibinfo{author}{Renner, R.}
\newblock \emph{\bibinfo{title}{{Security of Quantum Key Distribution}}}.
\newblock Ph.D. thesis, \bibinfo{school}{ETH Zurich} (\bibinfo{year}{2005}).
\newblock \urlprefix\url{http://arxiv.org/abs/quant-ph/0512258}.

\bibitem{Renner2005a}
\bibinfo{author}{Renner, R.}, \bibinfo{author}{Gisin, N.} \&
  \bibinfo{author}{Kraus, B.}
\newblock \bibinfo{title}{{Information-theoretic security proof for
  quantum-key-distribution protocols}}.
\newblock \emph{\bibinfo{journal}{Physical Review A}}
  \textbf{\bibinfo{volume}{72}}, \bibinfo{pages}{012332}
  (\bibinfo{year}{2005}).
\newblock \urlprefix\url{https://doi.org/10.1103/PhysRevA.72.012332}.

\bibitem{Watanabe2008}
\bibinfo{author}{Watanabe, S.}, \bibinfo{author}{Matsumoto, R.} \&
  \bibinfo{author}{Uyematsu, T.}
\newblock \bibinfo{title}{{Tomography increases key rates of
  quantum-key-distribution protocols}}.
\newblock \emph{\bibinfo{journal}{Physical Review A}}
  \textbf{\bibinfo{volume}{78}}, \bibinfo{pages}{042316}
  (\bibinfo{year}{2008}).
\newblock \urlprefix\url{https://doi.org/10.1103/PhysRevA.78.042316}.

\bibitem{Matsumoto2013}
\bibinfo{author}{Matsumoto, R.}
\newblock \bibinfo{title}{{Improved asymptotic key rate of the B92 protocol}}.
\newblock \emph{\bibinfo{journal}{IEEE International Symposium on Information
  Theory - Proceedings}} \bibinfo{pages}{351--353} (\bibinfo{year}{2013}).
\newblock \urlprefix\url{https://doi.org/10.1109/ISIT.2013.6620246}.

\bibitem{Vakhitov2001}
\bibinfo{author}{Vakhitov, A.}, \bibinfo{author}{Makarov, V.} \&
  \bibinfo{author}{Hjelme, D.~R.}
\newblock \bibinfo{title}{{Large pulse attack as a method of conventional
  optical eavesdropping in quantum cryptography}}.
\newblock \emph{\bibinfo{journal}{Journal of Modern Optics}}
  \textbf{\bibinfo{volume}{48}}, \bibinfo{pages}{2023--2038}
  (\bibinfo{year}{2001}).
\newblock \urlprefix\url{https://doi.org/10.1080/09500340108240904}.

\bibitem{Gisin2006}
\bibinfo{author}{Gisin, N.}, \bibinfo{author}{Fasel, S.},
  \bibinfo{author}{Kraus, B.}, \bibinfo{author}{Zbinden, H.} \&
  \bibinfo{author}{Ribordy, G.}
\newblock \bibinfo{title}{{Trojan-horse attacks on quantum-key-distribution
  systems}}.
\newblock \emph{\bibinfo{journal}{Physical Review A}}
  \textbf{\bibinfo{volume}{73}}, \bibinfo{pages}{022320}
  (\bibinfo{year}{2006}).
\newblock \urlprefix\url{https://doi.org/10.1103/PhysRevA.73.022320}.

\bibitem{Lucamarini2015}
\bibinfo{author}{Lucamarini, M.} \emph{et~al.}
\newblock \bibinfo{title}{{Practical Security Bounds Against the Trojan-Horse
  Attack in Quantum Key Distribution}}.
\newblock \emph{\bibinfo{journal}{Physical Review X}}
  \textbf{\bibinfo{volume}{5}}, \bibinfo{pages}{031030} (\bibinfo{year}{2015}).
\newblock \urlprefix\url{https://doi.org/10.1103/PhysRevX.5.031030}.

\bibitem{Huttner1995}
\bibinfo{author}{Huttner, B.}, \bibinfo{author}{Imoto, N.},
  \bibinfo{author}{Gisin, N.} \& \bibinfo{author}{Mor, T.}
\newblock \bibinfo{title}{{Quantum cryptography with coherent states}}.
\newblock \emph{\bibinfo{journal}{Physical Review A}}
  \textbf{\bibinfo{volume}{51}}, \bibinfo{pages}{1863--1869}
  (\bibinfo{year}{1995}).
\newblock \urlprefix\url{https://doi.org/10.1103/PhysRevA.51.1863}.

\bibitem{Lo2007}
\bibinfo{author}{Lo, H.} \& \bibinfo{author}{Preskill, J.}
\newblock \bibinfo{title}{{Security of quantum key distribution using weak
  coherent states with nonrandom phases}}.
\newblock \emph{\bibinfo{journal}{Quantum Information and Computation}}
  \textbf{\bibinfo{volume}{7}}, \bibinfo{pages}{431--458}
  (\bibinfo{year}{2007}).
\newblock \urlprefix\url{http://arxiv.org/abs/quant-ph/0610203}.

\bibitem{Fung2009}
\bibinfo{author}{Fung, C. C.-H.~F.}, \bibinfo{author}{Tamaki, K.},
  \bibinfo{author}{Qi, B.}, \bibinfo{author}{Lo, H.-K.~H.} \&
  \bibinfo{author}{Ma, X.}
\newblock \bibinfo{title}{{Security proof of quantum key distribution with
  detection efficiency mismatch}}.
\newblock \emph{\bibinfo{journal}{Quantum Inf.{\~{}}Comput.}}
  \textbf{\bibinfo{volume}{9}}, \bibinfo{pages}{0131--0165}
  (\bibinfo{year}{2009}).
\newblock \urlprefix\url{http://dl.acm.org/citation.cfm?id=2021264}.

\bibitem{Coles2016}
\bibinfo{author}{Coles, P.~J.}, \bibinfo{author}{Metodiev, E.~M.} \&
  \bibinfo{author}{L{\"{u}}tkenhaus, N.}
\newblock \bibinfo{title}{{Numerical approach for unstructured quantum key
  distribution}}.
\newblock \emph{\bibinfo{journal}{Nature Communications}}
  \textbf{\bibinfo{volume}{7}}, \bibinfo{pages}{11712} (\bibinfo{year}{2016}).
\newblock \urlprefix\url{https://doi.org/10.1038/ncomms11712}.

\bibitem{Boyd2010}
\bibinfo{author}{Boyd, S.} \& \bibinfo{author}{Vandenberghe, L.}
\newblock \emph{\bibinfo{title}{{Convex Optimization}}}
  (\bibinfo{publisher}{Cambridge University Press}, \bibinfo{year}{2004}).
\newblock \urlprefix\url{http://web.stanford.edu/{~}boyd/cvxbook/}.

\bibitem{Devetak2005}
\bibinfo{author}{Devetak, I.} \& \bibinfo{author}{Winter, A.}
\newblock \bibinfo{title}{{Distillation of secret key and entanglement from
  quantum states}}.
\newblock \emph{\bibinfo{journal}{Proceedings of the Royal Society A}}
  \textbf{\bibinfo{volume}{461}}, \bibinfo{pages}{207--235}
  (\bibinfo{year}{2005}).
\newblock \urlprefix\url{https://doi.org/10.1098/rspa.2004.1372}.

\bibitem{Renner2007}
\bibinfo{author}{Renner, R.}
\newblock \bibinfo{title}{{Symmetry of large physical systems implies
  independence of subsystems}}.
\newblock \emph{\bibinfo{journal}{Nature Physics}}
  \textbf{\bibinfo{volume}{3}}, \bibinfo{pages}{645--649}
  (\bibinfo{year}{2007}).
\newblock \urlprefix\url{https://doi.org/10.1038/nphys684}.

\bibitem{Berger2007}
\bibinfo{author}{Petersen, K.~B.} \& \bibinfo{author}{Pedersen, M.~S.}
\newblock \emph{\bibinfo{title}{{The Matrix Cookbook}}} (\bibinfo{year}{2012}).
\newblock \urlprefix\url{http://www2.imm.dtu.dk/pubdb/p.php?3274}.

\bibitem{Al-Mohy2012}
\bibinfo{author}{Al-Mohy, A.~H.} \& \bibinfo{author}{Higham, N.~J.}
\newblock \bibinfo{title}{{Improved Inverse Scaling and Squaring Algorithms for
  the Matrix Logarithm}}.
\newblock \emph{\bibinfo{journal}{SIAM Journal on Scientific Computing}}
  \textbf{\bibinfo{volume}{34}}, \bibinfo{pages}{C153--C169}
  (\bibinfo{year}{2012}).
\newblock \urlprefix\url{https://doi.org/10.1137/110852553}.

\bibitem{Frank1956}
\bibinfo{author}{Frank, M.} \& \bibinfo{author}{Wolfe, P.}
\newblock \bibinfo{title}{{An algorithm for quadratic programming}}.
\newblock \emph{\bibinfo{journal}{Naval Research Logistics Quarterly}}
  \textbf{\bibinfo{volume}{3}}, \bibinfo{pages}{95--110}
  (\bibinfo{year}{1956}).
\newblock \urlprefix\url{https://doi.org/10.1002/nav.3800030109}.

\bibitem{Bennett1984}
\bibinfo{author}{Bennett, C.~H.} \& \bibinfo{author}{Brassard, G.}
\newblock \bibinfo{title}{{Quantum cryptography: public key distribution and
  coin tossing}}.
\newblock In \emph{\bibinfo{booktitle}{International Conference on Computers,
  Systems {\&} Signal Processing, Bangalore, India}}, \bibinfo{pages}{175--179}
  (\bibinfo{year}{1984}).
\newblock \urlprefix\url{https://doi.org/10.1016/j.tcs.2011.08.039}.

\bibitem{Bennett1992}
\bibinfo{author}{Bennett, C.~H.}
\newblock \bibinfo{title}{{Quantum cryptography using any two nonorthogonal
  states}}.
\newblock \emph{\bibinfo{journal}{Physical Review Letters}}
  \textbf{\bibinfo{volume}{68}}, \bibinfo{pages}{3121--3124}
  (\bibinfo{year}{1992}).
\newblock \urlprefix\url{https://doi.org/10.1103/PhysRevLett.68.3121}.

\bibitem{Scarani2004}
\bibinfo{author}{Scarani, V.}, \bibinfo{author}{Ac{\'{i}}n, A.},
  \bibinfo{author}{Ribordy, G.} \& \bibinfo{author}{Gisin, N.}
\newblock \bibinfo{title}{{Quantum Cryptography Protocols Robust against Photon
  Number Splitting Attacks for Weak Laser Pulse Implementations}}.
\newblock \emph{\bibinfo{journal}{Physical Review Letters}}
  \textbf{\bibinfo{volume}{92}}, \bibinfo{pages}{057901}
  (\bibinfo{year}{2004}).
\newblock \urlprefix\url{https://doi.org/10.1103/PhysRevLett.92.057901}.

\bibitem{Bruss2002}
\bibinfo{author}{Bru{\ss}, D.}
\newblock \bibinfo{title}{{Optimal Eavesdropping in Quantum Cryptography with
  Six States}}.
\newblock \emph{\bibinfo{journal}{Physical Review Letters}}
  \textbf{\bibinfo{volume}{81}}, \bibinfo{pages}{3018--3021}
  (\bibinfo{year}{1998}).
\newblock \urlprefix\url{https://doi.org/10.1103/PhysRevLett.81.3018}.

\bibitem{Lo2005}
\bibinfo{author}{Lo, H.-K.}, \bibinfo{author}{Ma, X.} \& \bibinfo{author}{Chen,
  K.}
\newblock \bibinfo{title}{{Decoy State Quantum Key Distribution}}.
\newblock \emph{\bibinfo{journal}{Physical Review Letters}}
  \textbf{\bibinfo{volume}{94}}, \bibinfo{pages}{230504}
  (\bibinfo{year}{2005}).
\newblock \urlprefix\url{https://doi.org/10.1103/PhysRevLett.94.230504}.

\bibitem{Lo2012}
\bibinfo{author}{Lo, H.-K.}, \bibinfo{author}{Curty, M.} \&
  \bibinfo{author}{Qi, B.}
\newblock \bibinfo{title}{{Measurement-Device-Independent Quantum Key
  Distribution}}.
\newblock \emph{\bibinfo{journal}{Physical Review Letters}}
  \textbf{\bibinfo{volume}{108}}, \bibinfo{pages}{130503}
  (\bibinfo{year}{2012}).
\newblock \urlprefix\url{https://doi.org/10.1103/PhysRevLett.108.130503}.

\bibitem{Stacey2015}
\bibinfo{author}{Stacey, W.}, \bibinfo{author}{Annabestani, R.},
  \bibinfo{author}{Ma, X.} \& \bibinfo{author}{L{\"{u}}tkenhaus, N.}
\newblock \bibinfo{title}{{Security of quantum key distribution using a
  simplified trusted relay}}.
\newblock \emph{\bibinfo{journal}{Physical Review A}}
  \textbf{\bibinfo{volume}{91}}, \bibinfo{pages}{012338}
  (\bibinfo{year}{2015}).
\newblock \urlprefix\url{https://doi.org/10.1103/PhysRevA.91.012338}.

\bibitem{Bennett1992a}
\bibinfo{author}{Bennett, C.}, \bibinfo{author}{Brassard, G.} \&
  \bibinfo{author}{Mermin, N.}
\newblock \bibinfo{title}{{Quantum cryptography without Bell's theorem}}.
\newblock \emph{\bibinfo{journal}{Physical Review Letters}}
  \textbf{\bibinfo{volume}{68}}, \bibinfo{pages}{557--559}
  (\bibinfo{year}{1992}).
\newblock \urlprefix\url{https://doi.org/10.1103/PhysRevLett.68.557}.

\bibitem{Ferenczi2012}
\bibinfo{author}{Ferenczi, A.} \& \bibinfo{author}{L{\"{u}}tkenhaus, N.}
\newblock \bibinfo{title}{{Symmetries in quantum key distribution and the
  connection between optimal attacks and optimal cloning}}.
\newblock \emph{\bibinfo{journal}{Physical Review A}}
  \textbf{\bibinfo{volume}{85}}, \bibinfo{pages}{052310}
  (\bibinfo{year}{2012}).
\newblock \urlprefix\url{https://doi.org/10.1103/PhysRevA.85.052310}.

\bibitem{Coles2012}
\bibinfo{author}{Coles, P.~J.}
\newblock \bibinfo{title}{{Unification of different views of decoherence and
  discord}}.
\newblock \emph{\bibinfo{journal}{Physical Review A}}
  \textbf{\bibinfo{volume}{85}}, \bibinfo{pages}{042103}
  (\bibinfo{year}{2012}).
\newblock \urlprefix\url{https://doi.org/10.1103/PhysRevA.85.042103}.

\bibitem{Sajeed2015}
\bibinfo{author}{Sajeed, S.} \emph{et~al.}
\newblock \bibinfo{title}{{Security loophole in free-space quantum key
  distribution due to spatial-mode detector-efficiency mismatch}}.
\newblock \emph{\bibinfo{journal}{Physical Review A}}
  \textbf{\bibinfo{volume}{91}}, \bibinfo{pages}{062301}
  (\bibinfo{year}{2015}).
\newblock \urlprefix\url{https://doi.org/10.1103/PhysRevA.91.062301}.

\bibitem{Lydersen2010a}
\bibinfo{author}{Lydersen, L.} \emph{et~al.}
\newblock \bibinfo{title}{{Hacking commercial quantum cryptography systems by
  tailored bright illumination}}.
\newblock \emph{\bibinfo{journal}{Nature Photonics}}
  \textbf{\bibinfo{volume}{4}}, \bibinfo{pages}{686--689}
  (\bibinfo{year}{2010}).
\newblock \urlprefix\url{http://dx.doi.org/10.1038/nphoton.2010.214}.

\bibitem{Scarani2008}
\bibinfo{author}{Scarani, V.} \& \bibinfo{author}{Renner, R.}
\newblock \bibinfo{title}{{Quantum Cryptography with Finite Resources:
  Unconditional Security Bound for Discrete-Variable Protocols with One-Way
  Postprocessing}}.
\newblock \emph{\bibinfo{journal}{Physical Review Letters}}
  \textbf{\bibinfo{volume}{100}}, \bibinfo{pages}{200501}
  (\bibinfo{year}{2008}).
\newblock \urlprefix\url{https://doi.org/10.1103/PhysRevLett.100.200501}.

\bibitem{Sano2010}
\bibinfo{author}{Sano, Y.}, \bibinfo{author}{Matsumoto, R.} \&
  \bibinfo{author}{Uyematsu, T.}
\newblock \bibinfo{title}{{Secure key rate of the BB84 protocol using finite
  sample bits}}.
\newblock \emph{\bibinfo{journal}{IEEE International Symposium on Information
  Theory - Proceedings}} \textbf{\bibinfo{volume}{43}},
  \bibinfo{pages}{2677--2681} (\bibinfo{year}{2010}).
\newblock \urlprefix\url{https://doi.org/10.1109/ISIT.2010.5513653}.

\bibitem{Tomamichel2012a}
\bibinfo{author}{Tomamichel, M.}, \bibinfo{author}{Ci, C.},
  \bibinfo{author}{Lim, W.}, \bibinfo{author}{Gisin, N.} \&
  \bibinfo{author}{Renner, R.}
\newblock \bibinfo{title}{{Tight finite-key analysis for quantum
  cryptography}}.
\newblock \emph{\bibinfo{journal}{Nature Communications}}
  \textbf{\bibinfo{volume}{3}}, \bibinfo{pages}{634} (\bibinfo{year}{2012}).
\newblock \urlprefix\url{http://dx.doi.org/10.1038/ncomms1631}.

\bibitem{Coles2014b}
\bibinfo{author}{Coles, P.~J.}, \bibinfo{author}{Kaniewski, J.} \&
  \bibinfo{author}{Wehner, S.}
\newblock \bibinfo{title}{{Equivalence of wave--particle duality to entropic
  uncertainty}}.
\newblock \emph{\bibinfo{journal}{Nature Communications}}
  \textbf{\bibinfo{volume}{5}}, \bibinfo{pages}{5814} (\bibinfo{year}{2014}).
\newblock \urlprefix\url{http://dx.doi.org/10.1038/ncomms6814}.

\bibitem{Nielsen2000}
\bibinfo{author}{Nielsen, M.~A.} \& \bibinfo{author}{Chuang, I.}
\newblock \emph{\bibinfo{title}{{Quantum Computation and Quantum Information}}}
  (\bibinfo{publisher}{Cambridge University Press}, \bibinfo{year}{2000}).

\bibitem{Lo2004}
\bibinfo{author}{Lo, H.-K.}, \bibinfo{author}{Chau, H.} \&
  \bibinfo{author}{Ardehali, M.}
\newblock \bibinfo{title}{{Efficient Quantum Key Distribution Scheme and a
  Proof of Its Unconditional Security}}.
\newblock \emph{\bibinfo{journal}{Journal of Cryptology}}
  \textbf{\bibinfo{volume}{18}}, \bibinfo{pages}{133--165}
  (\bibinfo{year}{2005}).
\newblock \urlprefix\url{https://doi.org/10.1007/s00145-004-0142-y}.

\bibitem{Gittsovich2014}
\bibinfo{author}{Gittsovich, O.} \emph{et~al.}
\newblock \bibinfo{title}{{Squashing model for detectors and applications to
  quantum-key-distribution protocols}}.
\newblock \emph{\bibinfo{journal}{Physical Review A}}
  \textbf{\bibinfo{volume}{89}}, \bibinfo{pages}{012325}
  (\bibinfo{year}{2014}).
\newblock \urlprefix\url{https://doi.org/10.1103/PhysRevA.89.012325}.

\end{thebibliography}

\onecolumngrid
\appendix

\section{Proof of Theorem~\ref{thm1}}\label{sctAppProof}

\subsection{Standard form for semidefinite programs}\label{sctAppAsdps}

To prove Theorem~\ref{thm1}, we will make use of the so-called standard form for semidefinite programs (e.g., see page 265 of \cite{Boyd2010}), given as follows:
\begin{align}
\label{eqnSDPform}
&\text{\underline{Primal Problem}}			&\text{\underline{Dual Problem}}\\
&\min_X \hspace{5pt}\ipa{A}{X} 			&\max_{\vec{y}}\hspace{3pt}\vec{b} \cdot \vec{y}\hspace{21pt}\notag\\
&\text{subject to:}						&\text{subject to: }\hspace{10pt}\notag\\
&\ipa{B_1}{X} = b_1				&A \geq \vec{y}\cdot \vec{B}  \hspace{19pt}\notag\\
&\hspace{20pt}\vdots 					&\vec{y}\in \mathbb{R}^n \hspace{29pt}\notag\\
&\ipa{B_m}{X} = b_m\notag\\
&X\geq 0\notag
\end{align}
Here, the inner product is of the Hilbert-Schmidt form, $\ipa{A}{B}:= \Tr(A\ad B)$, and the vector notation $\vec{a}\cdot \vec{b}$ is shorthand for $\sum_{j=1}^m a_j b_j$.

If $\mu$ and $\nu$ are, respectively, the solutions of the primal and dual problems, then in general one has
\begin{align}
\label{eqnWeakDualityGeneral}
\mu \geq \nu\,,
\end{align}
which is known as weak duality \cite{Boyd2010}. Furthermore, under certain conditions, \eqref{eqnWeakDualityGeneral} turns into an equality $\mu = \nu$, known as strong duality. A sufficient condition for strong duality to hold is Slater's condition (e.g., page 265 of \cite{Boyd2010}), which can be stated as follows for the standard-form SDP in \eqref{eqnSDPform}.

\textbf{Slater's condition:} Let $\mathbf{X} = \{X \geq 0 : \ipa{B_j}{X} = b_j\,, \forall j\}$ be the set of feasible $X$ for the primal problem in \eqref{eqnSDPform}. Suppose that $\mathbf{X} \neq \emptyset$. Furthermore, suppose that there exists a $\vec{y} \in \mathbb{R}^n$ such that $A > \vec{y}\cdot \vec{B}$. Then
\begin{align}
\label{eqnSlatersConditionGeneral}
\mu = \nu\,.
\end{align}

\subsection{Inequality in \eqref{eq:DualLowerBound}}

In what follows, we first prove the inequality in \eqref{eq:DualLowerBound}, 
\begin{align}
\label{eqnEq8Repeated}
\alpha \geq \beta(\rho)\,,
\end{align}
and in the next subsection we establish the equality condition.

For some matrix $\sigma$, let $\vec{\sigma}:=\vectext (\sigma)$ denote the vectorization obtained by stacking the columns of $\sigma$. For two matrices $\sigma$ and $\tau$, note that the inner product between their vectorizations can be written as
\begin{align}
\label{eqninnerprodsigtau}
\vec{\sigma}\cdot \vec{\tau} = \Tr (\sigma^T \tau) \,,
\end{align}
where $^T$ is the transpose taken in the basis used to define the vectorization.

Now consider the gradient matrix $\nabla f(\rho)$ defined in \eqref{eqnGradDef836} and its vectorization $\vec{\nabla f}(\rho)$. Note that the quantity
\begin{align}
\label{eqngdef}
g(\rho,\sigma) &:= (\vec{\sigma} - \vec{\rho})\cdot  \vec{\nabla f}(\rho)\\
\label{eqnproofstep2}
&= \Tr \left[ (\sigma-\rho)^T\nabla f(\rho)\right]
\end{align}
quantifies how much the function $f$ changes whenever one moves from point $\rho$ to point $\sigma$. [Note that Eq.~\eqref{eqnproofstep2} rewrote $g(\rho,\sigma)$ in matrix notation, where $^T$ is the transpose in the same basis that is used to represent the gradient matrix.] More precisely, $g(\rho,\sigma)$ quantifies how much the linearization $L_{\rho}$ of $f$ changes, where $L_{\rho}$ is the linearization about point $\rho$. Specifically, the linearization is given by
\begin{align}
\label{eqlineardef}
L_{\rho}(\sigma) = f(\rho)+g(\rho,\sigma)\,.
\end{align}

The assumption of Thm.~\ref{thm1} is that $f$ is differentiable about the point $\rho$ of interest. Since $f$ is a convex, differentiable function over a convex set $\mathbf{S}$, the linearization $L_{\rho}$ always lies below the curve $f$ (e.g. page 69 of \cite{Boyd2010}). Hence, for any two points $\rho,\sigma \in \mathbf{S}$ we have 
\begin{align}
\label{eqnproofstep1}
f(\sigma) - f(\rho) & \geq g(\rho,\sigma)\,.
\end{align}

Now let $\rho^* \in \mathbf{S}$ minimize $f$ over $\mathbf{S}$. Then
\begin{align}
\label{eqnproofstep3}
f(\rho^*)  &\geq f(\rho)+ \Tr((\rho^*-\rho)^T\nabla f(\rho)) \\
\label{eqnproofstep4}
&\geq f(\rho)+ \min_{\sigma \in \mathbf{S}} \left[  \Tr((\sigma-\rho)^T\nabla f(\rho)) \right] \\
\label{eqnproofstep5}
&= f(\rho)- \Tr(\rho^T\nabla f(\rho)) + \min_{\sigma \in \mathbf{S}} \Tr(\sigma^T \nabla f(\rho))\,,
\end{align}
where \eqref{eqnproofstep4} exploits the fact that $\rho^* \in \mathbf{S}$. Hence finding a lower bound on $\alpha$ reduces to the minimization problem
\begin{align}
\label{eqnproofstep6}
\min_{\sigma \in \mathbf{S}} \Tr(\sigma^T \nabla f(\rho)).
\end{align}
This is a linear semidefinite program (SDP) and we may apply duality theory to obtain the dual SDP. In particular, our problem is essentially the standard SDP form given in Sec.~\ref{sctAppAsdps}, which gives the following dual problem
\begin{align}\label{eqnproofstep8}
\max_{\vec{y}\ \in \mathbf{S}^*(\rho)} \vec{\gamma} \cdot \vec{y}\,,
\end{align}
where
\begin{align}
\mathbf{S} &= \{\rho\in \mathbf{H_{+}} \mid \Tr(\Gamma_i \rho) = \gamma_i, \forall i\} \,, \\
\mathbf{S^*}(\sigma) &= \left\{ \vec{y} \in \mathbb{R}^n \mid \sum_i y_i \Gamma_i^T \leq \nabla f(\sigma) \right\}\,.
\end{align}

Weak duality from \eqref{eqnWeakDualityGeneral} implies that
\begin{align}
\label{eq:PrimalDualIneq}
\min_{\sigma \in \mathbf{S}} \Tr(\sigma^T \nabla f(\rho)) \geq \max_{\vec{y}\ \in \mathbf{S}^*(\rho)} \vec{\gamma} \cdot \vec{y}\,.
\end{align}
Inserting \eqref{eq:PrimalDualIneq} into \eqref{eqnproofstep5} gives the desired lower bound in \eqref{eq:DualLowerBound}.

\subsection{Equality in \eqref{eq:DualLowerBound}}\label{sctAppAequality}

With the inequality in Theorem~\ref{thm1} proven, we now show that equality in \eqref{eq:DualLowerBound} holds if $\rho$ satisfies $f(\rho) = f(\rho^*) = \alpha$.

First, consider the inequality in \eqref{eq:PrimalDualIneq}. Slater's condition stated in Sec.~\ref{sctAppAsdps} provides a sufficient criteria for strong duality to hold and, hence, for \eqref{eq:PrimalDualIneq} to be an equality. To satisfy this condition it is adequate to show that $\mathbf{S} \neq \emptyset$ and there exists $\vec{y} \in \mathbb{R}^n$ such that $\sum_i y_i \Gamma_i^T < \nabla f(\rho)$. Since the set of constraints $\{\Gamma_i\}$ correspond to a valid density matrix, it immediately follows that $\mathbf{S} \neq \emptyset$. Since density matrices are constrained to have trace one, without loss of generality we may take $\Gamma_1 = \id$ and $\gamma_1 = 1$. Thus, if $\lambda_{\min}$ is the smallest eigenvalue of $\nabla f(\rho)$, it follows that $(\lambda_{\min}-1) \Gamma_1^T < \nabla f(\rho)$. So $\vec{y} = (\lambda_{\min}-1,0,...,0)^T $ satisfies $\sum_i y_i \Gamma_i^T < \nabla f(\rho)$. With Slater's condition satisfied, strong duality holds and
\begin{align}
\label{eqnproofstep7}
\min_{\sigma \in \mathbf{S}} \Tr(\sigma^T \nabla f(\rho)) = \max_{\vec{y}\ \in \mathbf{S}^*(\rho)} \vec{\gamma} \cdot \vec{y}\,.
\end{align}

We now show that if $f(\rho) = f(\rho^*)$ then equality in \eqref{eq:DualLowerBound} holds. Suppose that $f(\rho) = f(\rho^*)$, then \eqref{eqnproofstep5} yields
\begin{align}
\label{eqnproofstep10}
\min_{\sigma \in \mathbf{S}} \Tr((\sigma-\rho)^T\nabla f(\rho)) \leq 0\,.
\end{align}

Next, we state a lemma that provides a bound in the opposite direction of \eqref{eqnproofstep10}.
\begin{lemma}
For a point $\rho$ that minimizes $f$ over $\mathbf{S}$,
\begin{align}
\label{eqnproofstep11}
\min_{\sigma \in \mathbf{S}} \Tr((\sigma-\rho)^T\nabla f(\rho)) \geq 0\,.
\end{align}
\begin{proof}
The Karush-Kuhn-Tucker (KKT) conditions (e.g. page 243 of \cite{Boyd2010}) provide necessary conditions for the optimality of $\rho$. For our problem they say that if $\rho$ is optimal, then there exists a pair $(\vec{\lambda},Z) \in \mathbb{R}^n\times\mathbf{H}$ where $\mathbf{H}$ denotes the set of $d \times d$ Hermitian matrices such that
\begin{align}
\nabla f(\rho) + \sum_i\lambda_i{\Gamma_i}^T - Z &= 0 \,, \label{eqn:kktderive} \\
\Tr(Z^T\rho) &= 0 \,, \label{eqn:kktrhoZrelation} \\
Z &\geq 0 \,.
\end{align}
Let $\sigma \in \mathbf{S}$ then
\begin{align}
\Tr((\sigma-\rho)^T\nabla f(\rho)) &= \Tr((\sigma-\rho)^T(-\sum_i\lambda_i{\Gamma_i}^T + Z)) \\
&= \sum_i\lambda_i\Tr(\Gamma_i \rho) - \sum_i\lambda_i\Tr(\Gamma_i \sigma) + \Tr(Z^T \sigma) - \Tr(Z^T \rho) \\
&= \Tr(Z^T\sigma) \,,
\end{align}
where we have used the definition of $\mathbf{S}$ and \eqref{eqn:kktrhoZrelation} in the last equality. Since $\sigma$ and $Z^T$ are both positive semidefinite it follows that the trace of their product is nonnegative. Hence,
\begin{align}
\Tr((\sigma-\rho)^T\nabla f(\rho)) \geq 0 \,,
\end{align}
and since $\sigma$ is arbitrary, the desired result follows.
\end{proof}
\end{lemma}

Combining \eqref{eqnproofstep10} and \eqref{eqnproofstep11} we must have
\begin{align}
\min_{\sigma \in \mathbf{S}} \Tr((\sigma-\rho)^T\nabla f(\rho)) = 0\,.
\end{align}
Consequently,
\begin{align}
\label{eqnproofstep12}
f(\rho^*) &= f(\rho) + \min_{\sigma \in \mathbf{S}} \Tr((\sigma-\rho)^T\nabla f(\rho)) \\
\label{eqnproofstep13}
&= f(\rho)-\Tr(\rho^T\nabla f(\rho))+\max_{\vec{y} \in \mathbf{S}^*(\rho)} \vec{\gamma} \cdot \vec{y} \\
&= \beta(\rho)\,.
\end{align}

\section{Proof of Lemma~\ref{lemma_gradient_exists}}\label{sctAppGradient}

As noted in \eqref{eqngradient1123}, the gradient of interest has the form
\begin{align}
\label{eqngradient2}
\left[ \nabla f_{\epsilon}(\rho) \right]^T = \GC_{\epsilon}^\dagger (\log \GC_{\epsilon}(\rho)) - \GC_{\epsilon}^\dagger (\log \ZC(\GC_{\epsilon}(\rho))) \,.
\end{align}
To show that this expression is well-defined, we simply need to show that logarithms can be evaluated, which would be true if $\GC_{\epsilon}(\rho)$ and $\ZC(\GC_{\epsilon}(\rho))$ are full-rank matrices.

Consider the first of these matrices
\begin{align}
\label{eqngradient3}
\GC_{\epsilon}(\rho) = (\DC_{\epsilon}\circ \GC)(\rho)  = (1-\epsilon)\GC( \rho ) + \epsilon \hspace{2 pt} \id / d'\,.
\end{align}
Since $\GC( \rho )\geq 0$, it follows that
\begin{align}
\label{eqngradient3}
\GC_{\epsilon}(\rho) \geq \epsilon \hspace{2 pt} \id / d' > 0\,.
\end{align}
In other words, $\GC_{\epsilon}(\rho)$ is full rank.

Now consider the second of these matrices
\begin{align}
\label{eqngradient4}
\ZC(\GC_{\epsilon}(\rho)) &= (1-\epsilon)\ZC(\GC( \rho )) + \epsilon \hspace{2 pt} \ZC(\id )/ d'\\
&= (1-\epsilon)\ZC(\GC( \rho )) + \epsilon \hspace{2 pt} \id / d'
\end{align}
where the second line notes that $\ZC$ is a pinching quantum channel and hence $\ZC(\id ) = \id$. Since $\ZC(\GC( \rho )) \geq 0$, we have
\begin{align}
\label{eqngradient3}
\ZC(\GC_{\epsilon}(\rho)) \geq \epsilon \hspace{2 pt} \id / d' > 0\,.
\end{align}
So both $\GC_{\epsilon}(\rho)$ and $\ZC(\GC_{\epsilon}(\rho))$ are full rank, implying that \eqref{eqngradient2} is well-defined.

\section{Proof of Theorem~\ref{thm2}}\label{appthm2proof}

\subsection{Continuity in $\rho$}\label{sctAppContinuity}

Here we state that the objective function $f(\rho)$ is continuous, which will be useful for our proof of Theorem~\ref{thm2}. 

First we state two helpful lemmas for the trace distance (or trace norm). The first lemma was proved in Ref.~\cite{Coles2014b}.
\begin{lemma}
\label{lemma1normPart1}
Let $\rho\geq 0$ and $\sigma \geq 0$ be positive semidefinite matrices. Then
\begin{align}
\label{eqn1normlemmaPart1}
\|\rho - \sigma \|_1^2 \leq (\Tr \rho +\Tr \sigma)^2 - 4 \| \sqrt{\rho} \sqrt{\sigma}    \|_1^2  \,,
\end{align}
where $\| A \|_1 = \Tr\sqrt{A\ad A}$ is the trace norm.
\end{lemma}

The next lemma is a corollary of Lemma~\ref{lemma1normPart1}.
\begin{lemma}
\label{lemma1norm}
Let $\rho\geq 0$ and $\sigma \geq 0$ be positive semidefinite matrices. Let $\EC$ be a completely positive trace-nonincreasing (CPTNI) map. Then
\begin{align}
\label{eqn1normlemma}
\|\rho - \sigma \|_1 \geq \| \EC(\rho) - \EC(\sigma) \|_1\,,
\end{align}
where $\| A \|_1 = \Tr\sqrt{A\ad A}$ is the trace norm.
\end{lemma}
\begin{proof}
Let $\rho - \sigma = Q - S$, where $Q\geq 0$ and $S\geq 0$, and $Q$ and $S$ have orthogonal support. Hence 
\begin{align}
\label{eqn1normlemmaPart12}\|\rho - \sigma \|_1 &= \|Q - S \|_1 \\
\label{eqn1normlemmaPart13}&=  \Tr Q +\Tr S \\
\label{eqn1normlemmaPart14}&\geq \Tr [\EC(Q)] +\Tr [\EC(S)] \\
\label{eqn1normlemmaPart15}&\geq \|\EC(Q) - \EC(S) \|_1 \\
\label{eqn1normlemmaPart16}&= \|\EC(\rho) - \EC(\sigma) \|_1\,.
\end{align}
Here, \eqref{eqn1normlemmaPart14} used the fact that $\EC$ is trace-nonincreasing. Also, \eqref{eqn1normlemmaPart15} follows directly from \eqref{eqn1normlemmaPart1}.
\end{proof}

Next we state a lemma, known as Fannes' inequality, that entropy is continuous. The proof can be found, e.g., in Ref.~\cite{Nielsen2000}. 
\begin{lemma}
\label{lemmaFannes}
Let $\rho\geq 0$ and $\sigma \geq 0$ be $n \times n$ positive semidefinite matrices, such that $\|\rho - \sigma\|_1 \leq \kappa \leq 1/e$. Then
\begin{align}
\label{eqn1normlemma2}
|H(\rho) - H(\sigma)| \leq \kappa \log (n /\kappa)\,.
\end{align}
\end{lemma}

Note that the right-hand side of \eqref{eqn1normlemma2} goes to zero as $\kappa\to 0$. Finally we state the continuity of our objective function.
\begin{lemma}
\label{lemmaContinuity}
Let $\rho$ and $\sigma$ be (normalized) density matrices such that $\|\rho - \sigma\|_1 \leq \kappa \leq 1/e$. Let $f(\rho)$ be defined as in \eqref{eqnPAEC4b},
\begin{align} \label{eqnPAEC4b_App}
f(\rho) = D(\hspace{2pt} \GC(\rho) || \ZC(\GC(\rho))\hspace{2pt})\,,
\end{align}
where $\GC$ is a completely positive trace-nonincreasing (CPTNI) map and $\ZC$ is a pinching quantum channel. Suppose $\GC(\rho)$ and $\GC(\sigma)$ are $d ' \times d '$. Then
\begin{align}
\label{eqn1normlemma3}
|f(\rho) - f(\sigma)| \leq 2\kappa \log (d' / \kappa)\,.
\end{align}
\end{lemma}
\begin{proof}
It is straightforward to show that
\begin{align} \label{eqn1normlemma4}
f(\rho) = H\left[\ZC(\GC(\rho))\right] - H\left[\GC(\rho)\right] \,.
\end{align}
Next note that Lemma~\ref{lemma1norm} implies that $\|\GC(\rho) - \GC(\sigma)\|_1 \leq \kappa$ and $\|\ZC(\GC(\rho)) - \ZC(\GC(\sigma))\|_1 \leq \kappa$. So from \eqref{eqn1normlemma4} we have
\begin{align} 
\label{eqn1normlemma5}
\big |f(\rho) - f(\sigma) \big | &=  \Big |  H \left[\ZC(\GC(\rho))\right] - H\left[ \GC(\rho)\right] - H \left[\ZC(\GC(\sigma))\right] + H\left[\GC(\sigma)\right] \Big |\\
\label{eqn1normlemma6}
  &\leq  \Big |  H \left[\ZC(\GC(\rho))\right] - H \left[\ZC(\GC(\sigma))\right] \Big | + \Big | H\left[ \GC(\rho)\right]  -  H\left[\GC(\sigma)\right]  \Big | \\
\label{eqn1normlemma7}
  &\leq  \kappa \log (d' /\kappa) + \kappa \log (d' /\kappa) \\
\label{eqn1normlemma8}
  &\leq  2\kappa \log (d' /\kappa)\,,
\end{align}
where \eqref{eqn1normlemma7} invoked Lemma~\ref{lemmaFannes}.
\end{proof}

\subsection{Continuity in $\epsilon$}\label{sctAppContinuityEpsilon}

Here we show another sort of continuity, namely, the continuity of $f_{\epsilon}(\rho)$ in $\epsilon$. First we need the following lemma.
\begin{lemma} \label{LemmaDelta1}
For any normalized density matrix $\rho$ and any completely positive trace-nonincreasing map (CPTNI) $\GC$, let $\GC_{\epsilon}(\rho) = \DC_{\epsilon}( \GC (\rho)) = (1-\epsilon)\GC (\rho) + \epsilon \id / d'$. Also, let $\ZC$ be a completely positive trace-preserving (CPTP) map. Then we have
\begin{align}
\label{eqnlemmadelta1}
\|\GC(\rho) - \GC_{\epsilon}(\rho)\|_1 &\leq  \epsilon (d' - 1)\,,\quad\text{and}\\
\label{eqnlemmadelta2}
\| \ZC(\GC(\rho)) - \ZC(\GC_{\epsilon}(\rho))\|_1 &\leq  \epsilon (d' -1)\,.
\end{align}
\begin{proof}
Note that \eqref{eqnlemmadelta2} follows from \eqref{eqnlemmadelta1} due to Lemma~\ref{lemma1norm}. So we just need to prove \eqref{eqnlemmadelta1}. Towards this end, we define
\begin{align}
\label{eqnlemmadelta12345a}
\Delta_{\epsilon}&:= \GC(\rho) - \GC_{\epsilon}(\rho) \\
\label{eqnlemmadelta12345b}
&= \GC(\rho) - ((1-\epsilon)*\GC(\rho)  +\epsilon \id /d')\\
\label{eqnlemmadelta12345c}
&= \epsilon*(\GC(\rho) -   \id /d' )=\epsilon \Delta
\end{align}
with $\Delta:=\GC(\rho) -   \id /d' $. Now note that $ 0 \leq \GC(\rho) \leq \id$, which implies
\begin{align}
\label{eqnlemmadelta111}
 -   \id /d' \leq \Delta  \leq (1-  1 / d') \id
\end{align}
Since $1/d' \leq (1-1/d')$, we have
\begin{align}
\label{eqnlemmadelta112}
 -(1-1/d')\id  \leq \Delta  \leq (1-  1 / d') \id\,.
\end{align}
This means that the $d'$ eigenvalues $\{\lambda_j\}$ of $\Delta$ each satisfy $| \lambda_j | \leq (1-  1 / d')$. Hence
\begin{align}
\label{eqnlemmadelta113}
\| \Delta \|_1  \leq d' -1\,.
\end{align}
Combining \eqref{eqnlemmadelta113} with \eqref{eqnlemmadelta12345c} proves \eqref{eqnlemmadelta1}.
\end{proof}
\end{lemma}

Now we prove our desired continuity statement.
\begin{lemma} \label{Lemmafepsiloncontinuity}
Let $\epsilon$ be such that $0 < \epsilon \leq 1/[e(d'-1)]$. Let $\rho$ be any density matrix. Then
\begin{align} \label{eqnfepsiloncont1}
\abs{ f(\rho) - f_{\epsilon}(\rho)} \leq \zeta_{\epsilon} \,,
\end{align}
where $\zeta_{\epsilon} := 2 \epsilon (d' - 1) \log \frac{d'}{\epsilon (d' - 1)} $.
\begin{proof}
We begin by expanding the expression of interest as follows
\begin{align} \label{eqnfepsiloncont2}
\abs{ f(\rho) - f_{\epsilon}(\rho)} &= \abs{H(\ZC(\GC(\rho))) - H(\GC(\rho)) + H(\GC_{\epsilon}(\rho)) - H(\ZC(\GC_{\epsilon}(\rho)))} \\
\label{eqnfepsiloncont3}
&\leq \abs{ H(\ZC(\GC(\rho))) - H(\ZC(\GC_{\epsilon}(\rho)))} + \abs{H(\GC(\rho)) - H(\GC_{\epsilon}(\rho)) } \,.
\end{align}
Next we exploit Lemma~\ref{LemmaDelta1} and Lemma~\ref{lemmaFannes}. In particular, applying Lemma~\ref{lemmaFannes} with $\kappa = \epsilon (d' - 1)$ and $n = d'$ gives
\begin{align} 
\label{eqnfepsiloncont3}
\abs{H(\GC(\rho)) - H(\GC_{\epsilon}(\rho)) } &\leq \kappa \log (n /\kappa) = \epsilon (d' - 1) \log \frac{d'}{\epsilon (d' - 1)}\\
\label{eqnfepsiloncont4}
\abs{ H(\ZC(\GC(\rho))) - H(\ZC(\GC_{\epsilon}(\rho)))} &\leq  \epsilon (d' - 1) \log \frac{d'}{\epsilon (d' - 1)}\,.
\end{align}
Finally, we obtain
\begin{align} 
\label{eqnfepsiloncont5}
\abs{ f(\rho) - f_{\epsilon}(\rho)} \leq 2 \epsilon (d' - 1) \log \frac{d'}{\epsilon (d' - 1)} = \zeta_{\epsilon}\,.
\end{align}
\end{proof}
\end{lemma}

\subsection{Proof of Theorem~\ref{thm2}}\label{sctAppThm2}

We first define a perturbed optimization problem. Let $\epsilon \in \mathbb{R}$ such that $0 < \epsilon < 1/(de)$. Then consider
\begin{align}
\alpha(\epsilon) := \min_{\rho \in \mathbf{S}} f_{\epsilon}(\rho)
\end{align}
where $f_{\epsilon}$ was defined in \eqref{eqnPerturbedfunction1} as
\begin{align}
f_{\epsilon}(\rho) =D(\GC_{\epsilon}(\rho) || \ZC( \GC_{\epsilon}(\rho)))\,.
\end{align}

Next we relate $\alpha(\epsilon)$ to $\alpha$ with the following lemma.
\begin{lemma} \label{thm2corollary}
Let $\epsilon$ be such that $0 < \epsilon \leq 1/[e(d'-1)]$. Then
\begin{align} \label{eqn:corollarythm2}
\abs{ \alpha - \alpha(\epsilon )} \leq \zeta_{\epsilon} \,,
\end{align}
where $\zeta_{\epsilon} := 2 \epsilon (d' - 1) \log \frac{d'}{\epsilon (d' - 1)} $.
\begin{proof}
Note that \eqref{eqn:corollarythm2} is equivalent to
\begin{align} \label{eqn:corollarythm22845}
- \zeta_\epsilon \leq \alpha - \alpha(\epsilon) \leq \zeta_\epsilon \,.
\end{align}

Now let $\alpha = f(\rho^*)$ and $\alpha(\epsilon) = f_{\epsilon}(\rho_{\epsilon}^*)$, where $\rho^* \in \mathbf{S}$ and $\rho_{\epsilon}^* \in \mathbf{S}$. We prove \eqref{eqn:corollarythm22845} in two steps. First we have
\begin{align}
\label{eqn:corollarythm226}
\alpha - \alpha(\epsilon) & = f(\rho^*) - f_{\epsilon}(\rho_{\epsilon}^*)\\
\label{eqn:corollarythm227}
&\leq  f(\rho_{\epsilon}^*) - f_{\epsilon}(\rho_{\epsilon}^*)\\
\label{eqn:corollarythm228}
&\leq  \zeta_\epsilon \,.
\end{align}
where the last line follows from Lemma~\ref{Lemmafepsiloncontinuity}.

Second we have
\begin{align}
\label{eqn:corollarythm229}
\alpha - \alpha(\epsilon) & = f(\rho^*) - f_{\epsilon}(\rho_{\epsilon}^*)\\
\label{eqn:corollarythm230}
&\geq  f(\rho^*) - f_{\epsilon}(\rho^*)\\
\label{eqn:corollarythm231}
&\geq -  \zeta_\epsilon \,.
\end{align}
where again the last line follows from Lemma~\ref{Lemmafepsiloncontinuity}. Taken together, \eqref{eqn:corollarythm228} and \eqref{eqn:corollarythm231} give the desired result. 
\end{proof}
\end{lemma}

Lemma~\ref{thm2corollary} connects $\alpha$ to $\alpha(\epsilon)$. So now we can proceed to lower bound $\alpha(\epsilon)$ to ultimately get a lower bound on $\alpha$.

The key observation is that $\alpha(\epsilon)$ corresponds precisely to $\alpha$ provided that one makes the substitution $\GC \to \GC_{\epsilon}$. Since this change simply involves substituting a different quantum channel for $\GC$, the lower bound that we previously found on $\alpha$ in Theorem~\ref{thm1} directly applies, with the appropriate substitution made. Recall that in Theorem~\ref{thm1} we found that, for any $\rho \in \mathbf{S}$ such that $\nabla f(\rho)$ exists,
\begin{align}
\label{eqnRecallTheorem1}
\alpha \geq \beta(\rho) = f(\rho)-\Tr(\rho^T \nabla f(\rho))+\max_{\vec{y} \in \mathbf{S}^*(\rho)} \vec{\gamma} \cdot \vec{y}\,,
\end{align}
where
\begin{align}
\label{eqnRecallTheorem12}
\mathbf{S^*}(\sigma) &:= \left\{ \vec{y} \in \mathbb{R}^n \mid \sum_i y_i \Gamma_i^T \leq \nabla f(\sigma) \right\}\,.
\end{align}

This result can be directly applied find a similar lower bound on $\alpha(\epsilon)$. One simply needs to substitute $\GC_{\epsilon}$ for $\GC$, which is equivalent to substituting $f_{\epsilon}$ for $f$. Hence we obtain
\begin{align}
\label{eqnalphaepsilonLB1}
\alpha(\epsilon) \geq \beta_{\epsilon}(\rho) := f_{\epsilon}(\rho)-\Tr(\rho^T \nabla f_{\epsilon}(\rho))+\max_{\vec{y} \in \mathbf{S}_{\epsilon}^*(\rho)} \vec{\gamma} \cdot \vec{y}\,,
\end{align}
where
\begin{align}
\label{eqnalphaepsilonLB2}
\mathbf{S}_{\epsilon}^*(\sigma) &:= \left\{ \vec{y} \in \mathbb{R}^n \mid \sum_i y_i \Gamma_i^T \leq \nabla f_{\epsilon}(\sigma) \right\}\,.
\end{align}
It is crucial to note that, while \eqref{eqnRecallTheorem1} applies only if $\nabla f(\rho)$ exists, \eqref{eqnalphaepsilonLB1} applies to all $\rho \in \mathbf{S}$ since $\nabla f_{\epsilon}(\rho)$ always exists.

Furthermore, we remark that the argument in Sec.~\ref{sctAppAequality} directly applies to $\alpha(\epsilon)$, so that we have
\begin{align}
\label{eqnalphaepsilonLB1equality}
\alpha(\epsilon) = \beta_{\epsilon}(\rho_{\epsilon}^*)
\end{align}
where $\rho_{\epsilon}^*$ is a state that satisfies $\alpha(\epsilon) = f_{\epsilon}(\rho_{\epsilon}^*)$.

Finally, we combine \eqref{eqnalphaepsilonLB1} with Lemma~\ref{thm2corollary}, specifically \eqref{eqn:corollarythm231}. This gives the following bound for any $\rho \in \mathbf{S}$:
\begin{align}
\alpha &\geq \alpha (\epsilon) - \zeta_{\epsilon}\\
&\geq \beta_{\epsilon}(\rho) - \zeta_{\epsilon}\,.
\end{align}

\section{Handling numerical imprecision}

\subsection{Imprecise representations}\label{impreciserepresentations}

As noted in the main text (Sec.~\ref{sctmainresult4}), it is impossible to provide exact floating-point representations of $\{\Gamma_i\}$ and $\{\gamma_i\}$. We address this impossibility by defining approximate representations, which we denote by $\{\tilde{\Gamma}_i\}$ and $\{\tilde{\gamma}_i\}$ respectively. Of course, we could define the set of states that satisfy the constraints associated with these approximate representations:
\begin{align}
\label{eqnStildeDef}
\tilde{\mathbf{S}}:= \{\rho\in \mathbf{H_{+}} \mid \Tr(\tilde{\Gamma}_i \rho) = \tilde{\gamma}_i, \forall i\}\,.
\end{align}
But as illustrated in Fig.~\ref{fgrImprecision}, $\tilde{\mathbf{S}}$ does not coincide with the true set of interest $\mathbf{S}$. To address this, we will expand the set $\tilde{\mathbf{S}}$ until it encompasses $\mathbf{S}$, as the set $\mathbf{S}_{\epsilon '}$ does in Fig.~\ref{fgrImprecision}.

We begin by relating the approximate and exact representations by
\begin{align}
\label{eqnimprecision1}
\tilde{\Gamma}_i=\Gamma_i+\delta\Gamma_i\quad \text{and}\quad\tilde{\gamma}_i=\gamma_i+\delta\gamma_i\,,
\end{align}
where
\begin{align}
\label{eqnimprecision2}
\norm{\delta\Gamma_i}_\text{HS} < \epsilon_1 \quad \text{and}\quad \abs{\delta\gamma_i} < \epsilon_2\,,
\end{align}
for all $i$. Here, the Hilbert-Schmidt norm is defined by $\norm{A}_\text{HS} := \sqrt{\Tr (A\ad A)}$. Determining the constants $\epsilon_1$ and $\epsilon_2$ is rather technical and depends heavily on how the approximate observables and expectation values are computed. However, if the mantissa of the underlying representation is increased in length (or arbitrary precision arithmetic is used) and appropriate numerical algorithms are applied, the approximate representations will become arbitrarily accurate. Hence, $\epsilon_1$ and $\epsilon_2$ may be made as small as needed.

We define the quantity
\begin{align}
\label{eqnimprecision3}
\epsilon_\text{rep} := \epsilon_1 + \epsilon_2\,,
\end{align}
which measures our overall uncertainty in the variable representation. 

We now prove a lemma that motivates our treatment of the imprecision. Namely, it allows us to upper bound the numerical constraint violation.
\begin{lemma}
\label{lemma2}
Let $\rho \in \mathbf{S}$ and let the definitions in \eqref{eqnimprecision1}, \eqref{eqnimprecision2}, and \eqref{eqnimprecision3} hold. Then
\begin{align}
\abs{\Tr(\tilde{\Gamma}_i\rho)-\tilde{\gamma}_i} < \epsilon_\text{rep} \,.
\end{align}
\begin{proof}
We can apply the triangle inequality, the Cauchy-Schwarz inequality and the fact that $\rho$ is a density matrix to obtain
\begin{align}
\abs{\Tr(\tilde{\Gamma}_i\rho)-\tilde{\gamma}_i} &= \abs{\Tr(\delta\Gamma_i\rho)-\delta\gamma_i} \\
&\leq \abs{\Tr(\delta\Gamma_i\rho)}+\abs{\delta\gamma_i} \\
&\leq \norm{\delta\Gamma_i}_\text{HS}\norm{\rho}_\text{HS}+\abs{\delta\gamma_i} \\
&\leq \norm{\delta\Gamma_i}_\text{HS}+\abs{\delta\gamma_i} \\
&< \epsilon_1 + \epsilon_2\,.
\end{align}
\end{proof}
\end{lemma}

Motivated by Lemma~\ref{lemma2}, we consider the set
\begin{align}
\label{eqnSepprimeDefAgain}
\St_{\epsilon'} := \left\{\rho \in \mathbf{H}_+ \middle | \abs{\Tr(\tilde{\Gamma}_i\rho)-\tilde{\gamma}_i} < \epsilon' , \forall i \right\}\,,
\end{align}
where it is clear that $\lim_{\epsilon' \to 0} \St_{\epsilon'} = \St$. Furthermore, if we choose $\epsilon' = \epsilon_\text{rep}$, we obtain the set
\begin{align}
\label{eqnSeprep}
\St_{\epsilon_\text{rep}}   \supseteq \mathbf{S} \,,
\end{align}
which follows directly from Lemma~\ref{lemma2}. Hence, $\St_{\epsilon_\text{rep}}$ encompasses both $\mathbf{S}$ and $\tilde{\mathbf{S}}$. So we could use $\St_{\epsilon_\text{rep}}$ as the basis for our optimization. However, we have not yet accounted imprecision in the numerical solver. As discussed in the next subsection, we may need to expand $\St_{\epsilon_\text{rep}}$ further to allow for solver imprecision.

\subsection{Imprecise solvers}

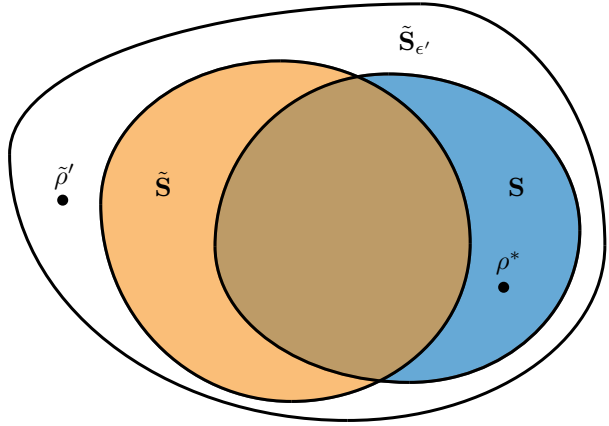
\begin{figure}[tbp]
\begin{center}
\vspace{0.2cm}
\begin{overpic}[width=8.1cm]{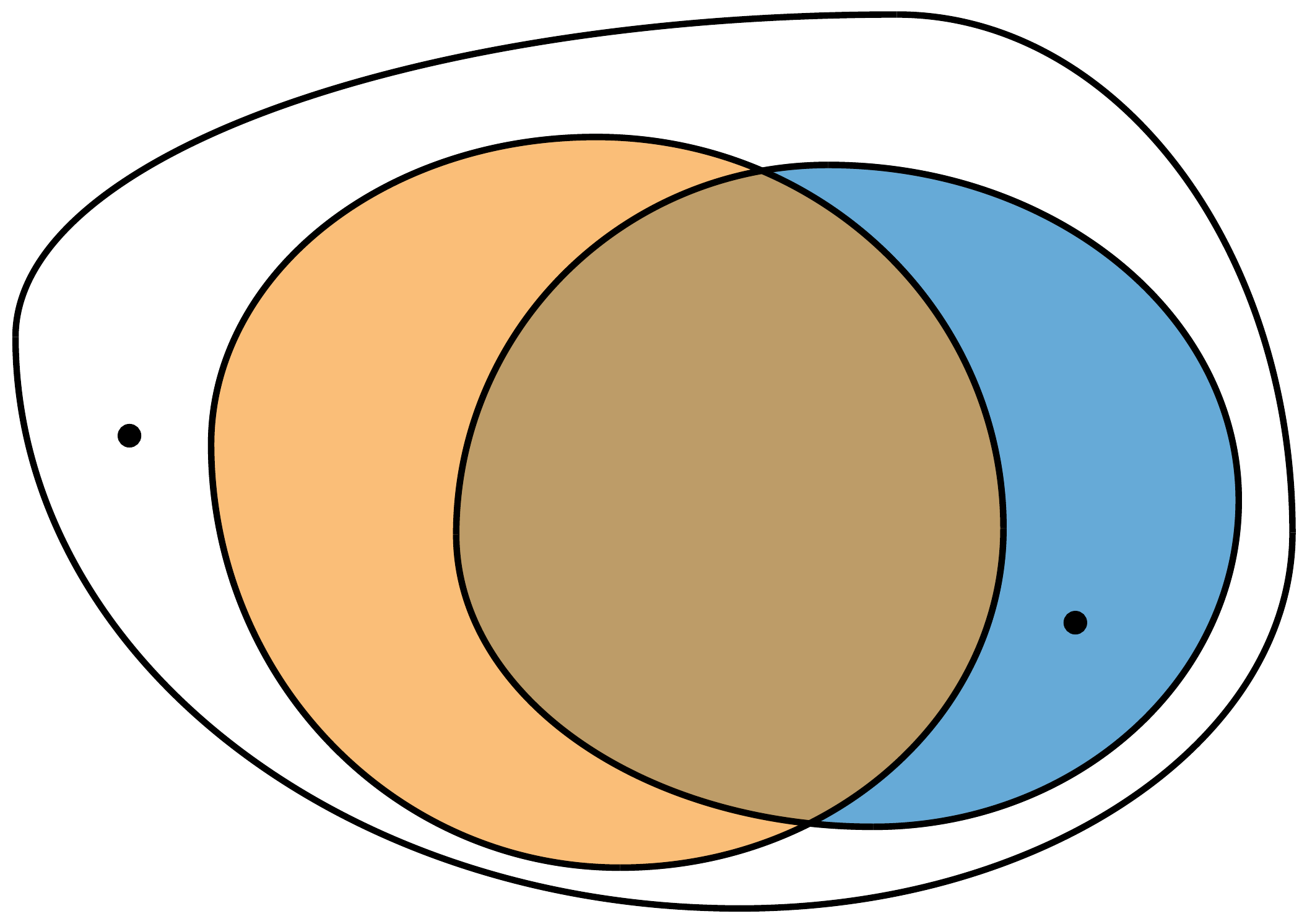}
\put(8.5,40){\normalsize $\tilde{\rho}'$}
\put(81,26){\normalsize $\rho^*$}
\put(83,38){\normalsize $\mathbf{S}$}
\put(25,38){\normalsize $\tilde{\mathbf{S}}$}
\put(65,62){\normalsize $\St_{\epsilon'}$}
\end{overpic}
\caption{Illustration of our method for handling numerical imprecision. Ideally we want to represent the set $\mathbf{S}$ in the computer. However, since the computer has finite precision, it stores $\tilde{\mathbf{S}}$ that approximates $\mathbf{S}$, where $\tilde{\mathbf{S}}$ is defined by \eqref{eqnStildeDef}. Furthermore, when we attempt to optimize over $\tilde{\mathbf{S}}$, the answer $\tilde{\rho}'$ returned by the computer may not even be in $\tilde{\mathbf{S}}$. We grow $\tilde{\mathbf{S}}$ into a new set $\St_{\epsilon'}$ that contains both $\mathbf{S}$ and $\tilde{\rho}'$, where $\tilde{\mathbf{S}}$ is defined in \eqref{eqnSepprimeDefAgain}. (Note that $\lim_{\epsilon ' \to 0} \St_{\epsilon'} = \tilde{\mathbf{S}}$.) With this expanded set $\St_{\epsilon'}$, we may then find a reliable lower bound by applying the technique discussed in the main text (Sec. \ref{sctmainresult4}).}
\label{fgrImprecision}
\end{center}
\end{figure}

Up to this point in our analysis, we have neglected the fact that no numerical solver is exact. In reality, the matrix $\tilde{\rho}$ returned by the solver may not be positive semidefinite or satisfy the approximate constraints defined by $\{\tilde{\Gamma}_i\}$ and $\{\tilde{\gamma}_i\}$. Here we present a method for addressing this situation.

First, we add a bit of the identity to $\tilde{\rho}$ to obtain a positive semidefinite matrix. Let $\lambda_\text{min}$ denote the smallest eigenvalue of $\tilde{\rho}$, then
\begin{align}
\tilde{\rho}' :=
\begin{cases}
\tilde{\rho} - \lambda_\text{min}\id \,, & \lambda_\text{min} < 0 \\
\tilde{\rho} \,, & \text{otherwise}
\end{cases}
\end{align}
is positive semidefinite. 

Second, we expand the set over which we optimize such that it encompasses $\tilde{\rho}'$, as illustrated in Fig.~\ref{fgrImprecision}. Let $\epsilon_\text{sol}$ be a positive real number such that
\begin{align} \label{eqnepsilonsolineq}
\abs{\Tr(\tilde{\Gamma}_i\tilde{\rho}')-\tilde{\gamma}_i} < \epsilon_\text{sol} \,
\end{align}
for all $i$. The quantity $\epsilon_\text{sol}$ describes how close $\tilde{\rho}'$ is to satisfying the approximate constraints provided that it is chosen to be as small as possible. The closer $\tilde{\rho}'$ is to satisfying the constraints, the smaller $\epsilon_\text{sol}$ will be.

We now have two quantities: $\epsilon_\text{rep}$ describes the representation precision and $\epsilon_\text{sol}$ describes the solver precision. We define the quantity
\begin{align}
\epsilon' = \max(\epsilon_\text{rep},\epsilon_\text{sol}) \,.
\end{align}
The reason for defining $\epsilon'$ is that it specifies how large of a set we need to contain both the imprecise solution returned by the solver and the exact solution. If we can construct such a set, then we may relax both Theorem~\ref{thm1} and Theorem~\ref{thm2} to obtain the computationally feasible Theorem~\ref{thm3}.

Recall the relaxed set of approximate density matrices defined in \eqref{eqnSepprimeDefAgain}, rewritten here for convenience:
\begin{align}
\St_{\epsilon'} = \left\{\rho \in \mathbf{H}_+ \middle | \abs{\Tr(\tilde{\Gamma}_i\rho)-\tilde{\gamma}_i} < \epsilon' , \forall i \right\}\notag \,.
\end{align}
From Lemma~\ref{lemma2}, it follows that $\mathbf{S} \subseteq \St_{\epsilon'}$, and \eqref{eqnepsilonsolineq} implies that $\tilde{\rho}' \in \St_{\epsilon'}$. This is illustrated in Fig.~\ref{fgrImprecision}. 

In summary, our method for handling numerical imprecision involves minimizing the objective function $f(\rho)$ over $\St_{\epsilon'}$. Because $\mathbf{S} \subseteq \St_{\epsilon'}$, our method will yield a reliable lower bound on $\alpha = \min_{\rho \in \mathbf{S}} f(\rho)$.

\subsection{Proof of Theorem~\ref{thm3}}\label{app:proofRelaxed}

Our proof strategy for Theorem~\ref{thm3} will be similar to that for Theorem~\ref{thm2}, in that we will consider a slightly perturbed optimization problem to ensure a well-defined gradient, and we will invoke continuity to argue that the perturbed problem is close to the problem of interest. The only difference is that now the problem of interest will allow for small violations of the constraints.

Of course, we are ultimately interested in lower bounding the quantity defined in \eqref{eqnPAEC4}, 
\begin{align}
\alpha &=  \min_{\rho \in \mathbf{S}} f(\rho) \,.
\end{align}
But due to numerical imprecision (discussed above), we must consider a slightly different problem,
\begin{align}
\alpha_{\epsilon'} &:=  \min_{\rho \in \St_{\epsilon'}} f(\rho) \,.
\end{align}
where $\St_{\epsilon'}$ was defined in \eqref{eqnrelaxedconstraints2}. Fortunately, we have that $\mathbf{S} \subseteq \St_{\epsilon'}$ (see previous subsection), which implies that
\begin{align}
\label{eqnalphaepprime}
\alpha \geq \alpha_{\epsilon'}\,.
\end{align}
Hence, any lower bound on $\alpha_{\epsilon'}$ is also a lower bound on $\alpha$. So let us now find a lower bound on $\alpha_{\epsilon'}$. 

As noted above, we will in fact consider a further perturbation on the optimization problem, to ensure a well-defined gradient. Again, this is analogous to what we did in Theorem~\ref{thm2}. This perturbed optimization problem has the form
\begin{align}
 \alpha_{\epsilon'} (\epsilon ) := \min_{\rho \in \St_{\epsilon'}   } f_{\epsilon}(\rho)\,,
\end{align}
where $f_{\epsilon}$ was defined in \eqref{eqnPerturbedfunction1}. Suppose that $\rho_{\epsilon \epsilon' }^* \in \St_{\epsilon'}$ achieves this optimization, i.e., $\alpha_{\epsilon'} (\epsilon ) = f_{\epsilon}(\rho_{\epsilon \epsilon' }^*)$. Next, we consider any $\rho  \in \St_{\epsilon'}$, and we invoke an argument identical to that used in Eqs.~\eqref{eqnproofstep3}-\eqref{eqnproofstep5} to write
\begin{align}\label{thm333ineq123}
\alpha_{\epsilon'} (\epsilon ) = f_{\epsilon }(\rho_{\epsilon \epsilon' }^*)   &\geq f_{\epsilon }(\rho )+ \Tr \left[ ( \rho_{\epsilon \epsilon' }^*  -  \rho  )^T\nabla f_{\epsilon }( \rho  ) \right] \\
\label{thm333ineq}
&\geq f_{\epsilon}(\rho )- \Tr\left[ \rho^T \nabla f_{\epsilon}(\rho )\right] + \min_{\sigma \in \St_{\epsilon'} } \Tr \left[\sigma^T \nabla f_{\epsilon}(\rho )\right]\,.
\end{align}

Next we rewrite the minimization in \eqref{thm333ineq} as a maximization via duality.
\begin{lemma}
\label{lemmaDoublePerturbation2}
Let $\rho  \in \St_{\epsilon'}$, then
\begin{align} \label{eqnlemmaDoublePerturbation2}
 \min_{\sigma \in \St_{\epsilon'}} \Tr \left[\sigma^T \nabla f_{\epsilon}(\rho )\right]  = \max_{\vec{x} \in \Sh_{\epsilon}^*(\rho)} \vec{v}(\vec{\tilde{\gamma}}, \epsilon') \cdot \vec{x} \,.
\end{align}
where $\vec{v}(\vec{\tilde{\gamma}}, \epsilon') : = (\vec{\tilde{\gamma}}+\epsilon') \oplus (-\vec{\tilde{\gamma}}+\epsilon')$ is a $2n$-dimensional vector, and
\begin{align} \label{eqnlemmaDoublePerturbation223}
\Sh_{\epsilon}^*(\sigma) &:= \left\{ \vec{x} \in \mathbb{R}^{2n} \mid \sum_{i=1}^n x_i (\overline{\Gamma}_i^+)^T + \sum_{i=1}^n x_{i+n} (\overline{\Gamma}_i^-)^T \leq \overline{\nabla f_{\epsilon}(\sigma)} \right\} \\
\overline{\nabla f_{\epsilon}(\sigma)} &:= \nabla f_{\epsilon}(\sigma) \oplus \mathbf{0} \oplus \mathbf{0} \\
\overline{\Gamma}_i^+ &:= \tilde{\Gamma}_i \oplus \mathbf{i} \oplus \mathbf{0}  \\
\overline{\Gamma}_i^- &:= -\tilde{\Gamma}_i \oplus \mathbf{0} \oplus  \mathbf{i}\,.
\end{align}
Here, $\mathbf{0}$ is the $n\times n$ matrix of zeros, and $\mathbf{i} = \dya{i}$ is the projector onto the $i$-th element of the standard basis in $n$-dimensions.
\begin{proof}
First we rewrite the minimization problem as
\begin{align}
\min_{\sigma \in \mathbf{H}} & \Tr(\sigma^T \nabla f_{\epsilon}(\rho)) \\
\text{s.t.} & \Tr(\tilde{\Gamma}_i\sigma) \leq \tilde{\gamma}_i + \epsilon' \notag \\
& \Tr(-\tilde{\Gamma}_i\sigma) \leq -\tilde{\gamma}_i + \epsilon' \notag \\
& \sigma \geq 0 \notag
\end{align}
Next, we utilize slack variables (e.g., see page 131 of \cite{Boyd2010}); the idea is to replace each inequality constraint with an equality constraint, and a nonnegativity constraint. Letting $\vec{a}$ and $\vec{b}$ denote slack variables, the problem becomes
\begin{align}
\min_{\sigma \in \mathbf{H}} & \Tr(\sigma^T \nabla f_{\epsilon}(\rho)) \\
\text{s.t.} & \Tr(\tilde{\Gamma}_i\sigma) + a_i = \tilde{\gamma}_i + \epsilon' \notag \\
& \Tr(-\tilde{\Gamma}_i\sigma) + b_i = -\tilde{\gamma}_i + \epsilon' \notag \\
& \sigma \geq 0 \notag \\
& \vec{a},\vec{b} \geq 0 \notag
\end{align}

Next we recast the problem so that there is again one positive semidefinite variable. First we define the following $n \times n$ diagonal matrices:
\begin{align}
\mathbf{a} = \text{diag}(\vec{a}), \quad \mathbf{b} = \text{diag}(\vec{b}), \quad \mathbf{0} = \text{diag}(\vec{0}),\quad \text{and }\mathbf{i} = \dya{i}\,,
\end{align}
where $\vec{0}$ is the $n$-dimensional zero vector, and $\dya{i}$ is the projector onto the $i$-th element of the standard $n$-dimensional basis. Next we define the block-diagonal matrices
\begin{align}
\overline{\sigma} &= \sigma \oplus \mathbf{a} \oplus \mathbf{b} \\
\overline{\nabla f_{\epsilon}(\rho)} &= \nabla f_{\epsilon}(\rho) \oplus \mathbf{0} \oplus \mathbf{0} \\
\overline{\Gamma}_i^+ &= \tilde{\Gamma}_i \oplus \mathbf{i} \oplus \mathbf{0}  \\
\overline{\Gamma}_i^- &= -\tilde{\Gamma}_i \oplus \mathbf{0} \oplus  \mathbf{i}\,.
\end{align}
The optimization problem then becomes
\begin{align}
\min_{\overline{\sigma} \in \mathbf{H}} & \Tr \left(\overline{\sigma}^T   \overline{\nabla f_{\epsilon}(\rho)} \right) \\
\text{s.t.} & \Tr \left(\overline{\Gamma}_i^\pm \overline{\sigma} \right) = \pm \tilde{\gamma}_i + \epsilon' \notag \\
& \overline{\sigma} \geq 0 \notag
\end{align}
This is a semidefinite program of a form identical to the one that appears in the proof of Theorem~\ref{thm1}. Duality theory, as discussed briefly in Sec.~\ref{sctAppAsdps}, yields the dual problem
\begin{align} \label{thm3dual}
\max_{\vec{x} \in \Sh_{\epsilon}^*(\rho)} \vec{v}(\vec{\tilde{\gamma}}, \epsilon') \cdot \vec{x} \,.
\end{align}
where $\vec{v}(\vec{\tilde{\gamma}}, \epsilon') : = (\vec{\tilde{\gamma}}+\epsilon') \oplus (-\vec{\tilde{\gamma}}+\epsilon')$ is a $2n$-dimensional vector.
\end{proof}
\end{lemma}

Next we rewrite the optimization problem in a simpler way, as follows.

\begin{lemma}
\label{lemma_rewrite_simpler}
Let $\rho  \in \St_{\epsilon'}$, then
\begin{align} \label{eqnlemmaDoublePerturbation2_2}
  \max_{\vec{x} \in \Sh_{\epsilon}^*(\rho)} \vec{v}(\vec{\tilde{\gamma}}, \epsilon') \cdot \vec{x}  =   \max_{ (\vec{y},\vec{z}) \in \St_{\epsilon}^*(\rho)}  \left( \vec{\tilde{\gamma}} \cdot \vec{y}  - \epsilon' \sum_{i=1}^n z_i \right) \,,
\end{align}
where
\begin{align}
\St_{\epsilon}^*(\sigma) := \left\{ (\vec{y},\vec{z}) \in (\mathbb{R}^{n}, \mathbb{R}^{n}) \mid -\vec{z} \leq \vec{y} \leq \vec{z} ,\hspace{4pt} \sum_{i=1}^n  y_i   (\tilde{\Gamma}_i )^T  \leq \nabla f_{\epsilon}(\sigma) \right\} \,.
\end{align}
\begin{proof}
First, note that we can rewrite the set $\Sh_{\epsilon}^*(\rho)$ as follows
\begin{align}
\label{eqn_Sh_set}
\Sh_{\epsilon}^*(\sigma) := \left\{ \vec{x} \in \mathbb{R}^{2n} \mid \vec{x} \leq 0,\hspace{4pt} \sum_{i=1}^n ( x_i - x_{i+n} )  (\tilde{\Gamma}_i )^T  \leq \nabla f_{\epsilon}(\sigma) \right\} \,.
\end{align}
This formula is derived by breaking up the constraints in \eqref{eqnlemmaDoublePerturbation223} associated with different blocks. The first block gives the constraint
\begin{align}
\sum_{i=1}^n ( x_i - x_{i+n} )  (\tilde{\Gamma}_i )^T  \leq \nabla f_{\epsilon}(\sigma).
\end{align}
The second block gives the constraint $\sum_i^n x_i \dya{i} \leq 0$, and the third block gives the constraint $\sum_i^n x_{i+n} \dya{i} \leq 0$. These latter two constraints imply that $x_i \leq 0$ for all $i \in [1,2n]$, or in other words, that $\vec{x} \leq 0$.

Next we write $\vec{x} = \vec{x}_1 \oplus \vec{x}_2$ where $\vec{x}_1 \in \mathbb{R}^{n}$ and $\vec{x}_2 \in \mathbb{R}^{n}$. Then we define $\vec{y} := \vec{x}_1 - \vec{x}_2$ and $\vec{z} := -\vec{x}_1 - \vec{x}_2$. With these definitions, we rewrite the objective function as
\begin{align}
\vec{v}(\vec{\tilde{\gamma}}, \epsilon') \cdot \vec{x} &= (\vec{\tilde{\gamma}}+\epsilon') \cdot \vec{x}_1 + (-\vec{\tilde{\gamma}}+\epsilon') \cdot \vec{x}_2 = \vec{\tilde{\gamma}} \cdot \vec{y}  - \epsilon' \sum_{i=1}^n z_i\,.
\end{align}
Next note that we can rewrite the constraint $\vec{x} \leq 0$ as
\begin{align}
-\vec{z} \leq \vec{y} \leq \vec{z}\,.
\end{align}
Hence we can define the set
\begin{align}
\St_{\epsilon}^*(\sigma) := \left\{ (\vec{y},\vec{z}) \in (\mathbb{R}^{n}, \mathbb{R}^{n}) \mid -\vec{z} \leq \vec{y} \leq \vec{z} ,\hspace{4pt} \sum_{i=1}^n  y_i   (\tilde{\Gamma}_i )^T  \leq \nabla f_{\epsilon}(\sigma) \right\} \,,
\end{align}
and the optimization problem becomes
\begin{align}
\max_{ (\vec{y},\vec{z}) \in \St_{\epsilon}^*(\rho)}  \left( \vec{\tilde{\gamma}} \cdot \vec{y}  - \epsilon' \sum_{i=1}^n z_i \right) \,.
\end{align}
\end{proof}
\end{lemma}

Combining \eqref{thm333ineq}, \eqref{eqnlemmaDoublePerturbation2}, and \eqref{eqnlemmaDoublePerturbation2_2} gives
\begin{align}\label{thm333ineq12233}
\alpha_{\epsilon'} (\epsilon ) = f_{\epsilon}(\rho_{\epsilon\epsilon' }^*)  \geq  f_{\epsilon}(\rho )- \Tr\left[ \rho^T \nabla f_{\epsilon}(\rho )\right]   +  \max_{ (\vec{y},\vec{z}) \in \St_{\epsilon}^*(\rho)} \left(  \vec{\tilde{\gamma}} \cdot \vec{y}  - \epsilon' \sum_{i=1}^n z_i  \right) \,.
\end{align}

Next we invoke continuity from Lemma~\ref{Lemmafepsiloncontinuity}, as we did in the proof of Theorem~\ref{thm2}, and we obtain the following inequality
\begin{align}
\label{eqn:thm3lowerboundineqoptimals67}
\alpha_{\epsilon '} - \alpha_{\epsilon '}(\epsilon) \geq f(\rho_{\epsilon\epsilon' }^*) - f_{\epsilon}(\rho_{\epsilon\epsilon' }^*) \geq - \zeta_{\epsilon} \,.
\end{align}

Finally, combining \eqref{eqnalphaepprime}, \eqref{thm333ineq12233}, and \eqref{eqn:thm3lowerboundineqoptimals67} gives the desired lower bound on $\alpha$ stated in \eqref{eqnthm3mainresult}, namely
\begin{align}
\label{eqn:thm3lowerboundineqoptimals67132}
\alpha &\geq \alpha_{\epsilon '}(\epsilon) - \zeta_{\epsilon}\\
 &\geq   f_{\epsilon}(\rho )- \Tr\left[ \rho^T \nabla f_{\epsilon}(\rho )\right]   +  \max_{ (\vec{y},\vec{z}) \in \St_{\epsilon}^*(\rho)}  \left( \vec{\tilde{\gamma}} \cdot \vec{y}  - \epsilon' \sum_{i=1}^n z_i \right) - \zeta_{\epsilon}\,.
\end{align}

\section{Tightness}\label{sctAppTightness}

In this section we show that the lower bound in Theorem~\ref{thm3} is tight. We want to show that as $\epsilon \to 0$ and $\epsilon' \to 0$, the minimizer $\rho^*_{\epsilon\epsilon'}$ of $f_{\epsilon}$ over $\St_{\epsilon'}$ satisfies
\begin{align} \label{eqnlimittoshowtight}
\beta_{\epsilon\epsilon'}(\rho^*_{\epsilon\epsilon'})-\zeta_\epsilon \to \alpha
\end{align}
where the left-hand side of \eqref{eqnlimittoshowtight} is the lower bound produced by Theorem~\ref{thm3}. We state this as the following proposition.

\begin{proposition}
Let $\rho^*_{\epsilon\epsilon'}$ be the minimizer of $f_{\epsilon}$ over $\St_{\epsilon'}$. Then
\begin{align}
\lim_{\epsilon,\epsilon' \to 0} \beta_{\epsilon\epsilon'}(\rho^*_{\epsilon\epsilon'})-\zeta_\epsilon = \alpha \,.
\end{align}

\begin{proof}
Taking the limit of the lower bound in Theorem~\ref{thm3} as $\epsilon' \to 0$ we recover the bound in Theorem~\ref{thm2}. That is
\begin{align}
\label{eqnlimittoshowtight1234}
\lim_{\epsilon,\epsilon' \to 0} \beta_{\epsilon\epsilon'}(\rho^*_{\epsilon\epsilon'})-\zeta_\epsilon = \lim_{\epsilon \to 0} \beta_{\epsilon} ( \rho^*_{\epsilon} )-\zeta_\epsilon \,,
\end{align}
where
\begin{align}
\rho^*_{\epsilon\epsilon'} := \arg \min_{\rho \in \St_{\epsilon'}} f_{\epsilon}(\rho)\\
\rho^*_{\epsilon} := \arg \min_{\rho \in \mathbf{S}} f_{\epsilon}(\rho)\,,
\end{align}
and in \eqref{eqnlimittoshowtight1234} we used the fact that $\lim_{\epsilon' \to 0} \rho^*_{\epsilon\epsilon'} = \rho^*_{\epsilon}$.

Now by \eqref{eqnalphaepsilonLB1equality}
\begin{align}
\lim_{\epsilon \to 0} \beta_{\epsilon} ( \rho^*_{\epsilon} )-\zeta_\epsilon &= \lim_{\epsilon \to 0} \alpha(\epsilon)  - \zeta_\epsilon \,.
\end{align}
This gives
\begin{align}
\lim_{\epsilon,\epsilon' \to 0} \beta_{\epsilon\epsilon'}(\rho^*_{\epsilon\epsilon'})-\zeta_\epsilon &= \lim_{\epsilon \to 0} \alpha(\epsilon)  - \zeta_\epsilon \,.
\end{align}
Finally, Lemma~\ref{thm2corollary} implies that
\begin{align}
\lim_{\epsilon,\epsilon' \to 0} \beta_{\epsilon\epsilon'}(\rho^*_{\epsilon\epsilon'})-\zeta_\epsilon = \alpha \,.
\end{align}
\end{proof}

\end{proposition}

Thus, given suitably small $\epsilon,\epsilon'$ it follows that the bound produced by Theorem~\ref{thm3} applied to a near-optimal state is arbitrarily tight.

\section{Examples}\label{sctAppExamples}

\subsection{Efficiency Mismatch}\label{sctAppExamples1}

For the BB84 protocol with detector efficiency mismatch, we model it as an entanglement-based protocol. We write Alice's POVM as
\begin{align}
P^A_1 = p_z \dya{0}\,,\quad P^A_2 = p_z \dya{1}\,,\quad P^A_3 = (1-p_z) \dya{+}\,,\quad P^A_4 = (1-p_z) \dya{-}
\end{align}
where $\{\ket{0},\ket{1}\}$ is the $z$-basis on a qubit, and $\ket{\pm} = 1/\sqrt{2}(\ket{0}\pm \ket{1})$. Here $p_z$ denotes the probability for Alice to measure in the $z$-basis. For our numerics, we chose $p_z \approx 1$, corresponding to using the $z$-basis most of the time.

We model Bob's system as a qutrit, where the one-photon subspace is modeled as a qubit subspace, and the third dimension is the vacuum. This third dimension is incorporated because of detector inefficiency, which may cause a no-click event. Bob's POVM elements associated with detecting a photon are given by
\begin{align}
P^B_1 = p_z \dya{0}\oplus 0\,,\quad P^B_2 = p_z \eta \dya{1}\oplus 0\,,\quad P^B_3 = (1-p_z) \dya{+}\oplus 0\,,\quad P^B_4 = (1-p_z) \eta\dya{-}\oplus 0\,,
\end{align}
where the direct sum is used here to embed qubit operators inside a qutrit Hilbert space. Note that $P^B_2$ and $P^B_4$ have a factor of $\eta$ due to detector inefficiency. (We assume one detector has perfect efficiency, while the other has efficiency $\eta$.) The last of Bob's POVM elements corresponds to a no-click event,
\begin{align}
P^B_5 = \id - \sum_{j=1}^4 P^B_j\,.
\end{align}

To generate the constraints in \eqref{eqnconstraints}, we simulate the data using a depolarizing channel with depolarizing probability~$p$,
\begin{align} \label{eqndepolarizingchannel}
\EC_{\text{dep}}(\rho) = (1-p)\rho + p \id / 2\,.
\end{align}
We consider the bipartite state generated from applying this channel to half of a maximally entangled state $\ket{\Phi}$,
\begin{align}
\rho_{AB}^{\text{sim}} = (\IC\ot \EC_{\text{dep}})(\dya{\Phi})\,,
\end{align}
and we emphasize that this state is only used to simulate experimental data. To obtain the constraints in \eqref{eqnconstraints}, we compute
\begin{align}
p_{jk} = \Tr ((P^A_j \ot P^B_k) \rho_{AB}^{\text{sim}} )\,.
\end{align}

Now consider Alice's and Bob's announcements. For sifting purposes, Alice announces her basis, and so \eqref{eqnkrausalice} becomes  
\begin{align} \label{eqnkrausaliceEM}
K^A_0 &=    \sqrt{ P^A_{1 }} \ot \ket{0}_{\At} \ot \ket{0}_{\Ab} + \sqrt{ P^A_{2 }} \ot \ket{0}_{\At} \ot \ket{1}_{\Ab}\\
K^A_1 &=    \sqrt{ P^A_{3 }} \ot \ket{1}_{\At} \ot \ket{0}_{\Ab} + \sqrt{ P^A_{4 }} \ot \ket{1}_{\At} \ot \ket{1}_{\Ab}\,.
\end{align}
Likewise Bob announces his basis and he announces whether he got a click or not. We can model this with three Kraus operators as follows
\begin{align} \label{eqnkrausbobEM}
K^B_0 &=    \sqrt{ P^B_{1 }} \ot \ket{0}_{\Bt} \ot \ket{0}_{\Bb} + \sqrt{ P^B_{2 }} \ot \ket{0}_{\Bt} \ot \ket{1}_{\Bb}\\
K^B_1 &=    \sqrt{ P^B_{3 }} \ot \ket{1}_{\Bt} \ot \ket{0}_{\Bb} + \sqrt{ P^B_{4 }} \ot \ket{1}_{\Bt} \ot \ket{1}_{\Bb}\\
K^B_2 &=    \sqrt{ P^B_{5 }} \ot \ket{2}_{\Bt} \ot \ket{0}_{\Bb} \,.
\end{align}
Next we consider the post-selection. Events where Bob does not receive a click, or where Alice and Bob use different bases, are discarded. Hence, \eqref{eqnstateafterannounce22} becomes
\begin{align} \label{eqnstateafterannounce22EM}
 \Pi =  \dya{0}_{\At} \ot \dya{0}_{\Bt}+ \dya{1}_{\At} \ot \dya{1}_{\Bt}   \,.
\end{align}

Finally, consider the isometry $V$ associated with the key map defined in \eqref{eqnstateafterannounce44}. We can define the key map such that Alice stores $0$ ($1$) in her key when she obtains outcome $P^A_1$ or $P^A_3$ ($P^A_2$ or $P^A_4$). This gives 
\begin{align} \label{eqnstateafterannounce44EM}
V = \ket{0}_R \ot \dya{0}_{\Ab} + \ket{1}_R \ot \dya{1}_{\Ab}\,,
\end{align}
with identity acting on all other subsystems. The above expressions allow one to define $\GC$ in \eqref{eqnfrho235431}, and hence define the optimization problem.

\subsection{Trojan-horse attack}\label{sctAppExamples2}

We model the BB84 protocol under a Trojan-horse attack as a prepare-and-measure protocol with sifting. As discussed in Sec.~\ref{sctqkdframework}, we treat this by constructing the source-replacement state,
\begin{align}
\ket{\psi}_{AA'} = \sqrt{\frac{p_z}{2}}\ket{0}\ket{\phi_{z+}} + \sqrt{\frac{p_z}{2}}\ket{1}\ket{\phi_{z-}} + \sqrt{\frac{1-p_z}{2}}\ket{2}\ket{\phi_{x+}} + \sqrt{\frac{1-p_z}{2}}\ket{3}\ket{\phi_{x-}}
\end{align}
where $\{\ket{\phi_{z\pm}},\ket{\phi_{x\pm}}\}$ are the signal states specified in \eqref{eqntrojansignal1}-\eqref{eqntrojansignal4}. For high-efficiency sifting~\cite{Lo2004}, we bias the probability distribution so that the $z$-basis is used most of the time, i.e., $p_z \approx 1$. Within this framework, Alice prepares her signal states by acting with a POVM on register system $A$, with POVM elements
\begin{align}
P^A_1 = \dya{0}\,,\quad P^A_2 = \dya{1}\,,\quad P^A_3 = \dya{2}\,,\quad P^A_4 = \dya{3} \,.
\end{align}
Bob measures in either the $z$- or $x$-basis via the following POVM
\begin{align}
P^B_1 = p_z\dya{z_+}\,,\quad P^B_2 = p_z\dya{z_-}\,,\quad P^B_3 = (1-p_z)\dya{x_+}\,,\quad P^B_4 = (1-p_z)\dya{x_-} \,,
\end{align}
where for simplicity we set the $p_z$ appearing in Bob's measurement to be the same value as that used for Alice's signal states.

Next we consider data simulation for the purpose of formulating the constraints in \eqref{eqnconstraints}. We model Eve's attack as a depolarizing channel \eqref{eqndepolarizingchannel} with depolarizing probability $p$. (Note that $p = 2Q$, where $Q$ is the error rate plotted in Fig.~\ref{fgrTrojan}.) Applying this channel to the state $\ket{\psi}_{AA'}$ gives
\begin{align}
\rho_{AB}^{\text{sim}} = (\IC\ot \EC_{\text{dep}})(\dya{\psi}_{AA'})\,.
\end{align}
To obtain the constraints in \eqref{eqnconstraints}, we compute
\begin{align} \label{eqntrojanconstraint1}
p_{jk} = \Tr ((P^A_j \ot P^B_k) \rho_{AB}^{\text{sim}} )\,.
\end{align}
Since Alice's density operator is fixed, we add the additional constraints specified by \eqref{eqnconstraints22}.

Now we consider the announcements made by Alice and Bob. Alice announces her choice of basis, so \eqref{eqnkrausalice} becomes
\begin{align}
K^A_0 &=    \sqrt{ P^A_{1 }} \ot \ket{0}_{\At} \ot \ket{0}_{\Ab} + \sqrt{ P^A_{2 }} \ot \ket{0}_{\At} \ot \ket{1}_{\Ab}\\
K^A_1 &=    \sqrt{ P^A_{3 }} \ot \ket{1}_{\At} \ot \ket{0}_{\Ab} + \sqrt{ P^A_{4 }} \ot \ket{1}_{\At} \ot \ket{1}_{\Ab}\,.
\end{align}
(We remark that, in this case, introducing the additional register system $\Ab$ is redundant since the key information can be read off directly from system $A$, but we do it here for completeness.)

Similarly, Bob announces his choice of basis, so \eqref{eqnkrausbob} becomes
\begin{align}
K^B_0 &= \sqrt{ P^B_{1 }} \ot \ket{0}_{\Bt} \ot \ket{0}_{\Bb} + \sqrt{ P^B_{2 }} \ot \ket{0}_{\Bt} \ot \ket{1}_{\Bb}\\
K^B_1 &= \sqrt{ P^B_{3 }} \ot \ket{1}_{\Bt} \ot \ket{0}_{\Bb} + \sqrt{ P^B_{4 }} \ot \ket{1}_{\Bt} \ot \ket{1}_{\Bb}\,.
\end{align}

For the post-selection, Alice and Bob discard events where they measure in different bases. So \eqref{eqnstateafterannounce22} becomes
\begin{align}
\Pi =  \dya{0}_{\At} \ot \dya{0}_{\Bt}+ \dya{1}_{\At} \ot \dya{1}_{\Bt}   \,.
\end{align}

Finally, consider the isometry $V$ associated with the key map defined in \eqref{eqnstateafterannounce44}. We can define the key map such that Alice stores $0$ ($1$) in her key when she obtains outcome $P^A_1$ or $P^A_3$ ($P^A_2$ or $P^A_4$). This gives 
\begin{align}
V = \ket{0}_R \ot \dya{0}_{\Ab} + \ket{1}_R \ot \dya{1}_{\Ab}\,,
\end{align}
with identity acting on all other subsystems. The above expressions allow one to define $\GC$ in \eqref{eqnfrho235431}, and hence define the optimization problem.

\subsection{BB84 protocol with phase-coherent signal states}\label{sctAppExamples3}

We model the BB84 protocol with phase-coherent signal states as a prepare-and-measure protocol with sifting, similar to how we modeled the Trojan-horse attack above. We apply the source-replacement scheme as described in Sec.~\ref{sctqkdframework}, with the state $\ket{\psi}_{AA'}$ from~\eqref{eqnSourceReplace94} given by 
\begin{align}
\ket{\psi}_{AA'} = \sqrt{\frac{p_z}{2}}\ket{0}\ket{\phi_{z+}} + \sqrt{\frac{p_z}{2}}\ket{1}\ket{\phi_{z-}} + \sqrt{\frac{1-p_z}{2}}\ket{2}\ket{\phi_{x+}} + \sqrt{\frac{1-p_z}{2}}\ket{3}\ket{\phi_{x-}}
\end{align}
where $\{\ket{\phi_{z\pm}},\ket{\phi_{x\pm}}\}$ are specified in \eqref{eqnlpsignal1}-\eqref{eqnlpsignal4}. Here, $p_z$ denotes the probability of Alice preparing a state in the $z$-basis, and it is biased to be close to one.

Alice's POVM acts on the register system $A$ with the standard basis elements, namely
\begin{align}
P^A_1 = \dya{0}\,,\quad P^A_2 = \dya{1}\,,\quad P^A_3 = \dya{2}\,,\quad P^A_4 = \dya{3} \,.
\end{align}

By applying a squashing model \cite{Gittsovich2014}, we model Bob's system as a qutrit, where the one-photon subspace is modeled as a qubit subspace, and the third dimension is the vacuum. This third dimension is incorporated because of channel loss, which may cause a no-click event. Bob's POVM elements are then 
\begin{align}
&P^B_1 = p_z\dya{0} \oplus 0 \,,\quad P^B_2 = p_z\dya{1} \oplus 0 \,,\quad P^B_3 = (1-p_z)\dya{+} \oplus 0 \,,\notag\\
& P^B_4 = (1-p_z)\dya{-} \oplus 0 \,,\quad P^B_5 = \id - \sum_{i=1}^4 P^B_i \,.
\end{align}

Next we consider data simulation for the purpose of formulating the constraints in \eqref{eqnconstraints}. We model the channel between Alice and Bob as a lossy channel $\EC_{\text{loss}}(\rho)$ with transmission probability $\eta$. Note that the action of this channel on a coherent state is $\ket{\alpha} \to \ket{\sqrt{\eta}\alpha}$. We can apply the channel to the state $\ket{\psi}_{AA'}$ to obtain the bipartite state
\begin{align}
\rho_{AB}^{\text{sim}} = (\IC\ot \EC_{\text{loss}})(\dya{\psi}_{AA'})\,.
\end{align}
The remainder of the model is identical to that of Sec.~\ref{sctAppExamples2} from \eqref{eqntrojanconstraint1} onwards. This defines the optimization problem.

\end{document}